\documentclass[11pt,reqno]{amsart}

\usepackage{epsfig}
\usepackage{amsmath, amsthm, amsfonts}
\usepackage{graphicx}

\numberwithin{equation}{section}

\newcommand{\R}{{\mathbb R}}
\newcommand{\C}{{\mathbb C}}
\newcommand{\N}{{\mathbb N}}

\newcommand{\Z}{{\mathbb Z}}

\renewcommand{\d}{\partial}
\renewcommand{\Re}{{\operatorname{Re\,}}}
\renewcommand{\Im}{{\operatorname{Im\,}}}

\newcommand{\Ai}{{\operatorname{Ai}}}

\newcommand{\al}{\alpha}

\newcommand{\ga}{\gamma}

\newcommand{\Ga}{\Gamma}

\newcommand{\la}{\lambda}

\newcommand{\ep}{\varepsilon}
\newcommand{\de}{\delta}
\newcommand{\De}{\Delta}

\newcommand{\sg}{\sigma}
\newcommand{\Sg}{\Sigma}

\newcommand{\Om}{\Omega}
\renewcommand{\th}{\theta}
\newcommand{\z}{\zeta}

\newcommand{\fcal}{{\mathcal F}}
\newtheorem{theo}{{\sc \bf Theorem}}[section]

\newtheorem{lem}[theo]{{\sc \bf Lemma}}
\newtheorem{prop}[theo]{{\sc \bf Proposition}}

\begin{document}

\title[Nonintersecting Brownian motions on the half-line]{Nonintersecting Brownian motions on the half-line and discrete Gaussian orthogonal polynomials}

\author{Karl Liechty}
\address{Department of Mathematics,
University of Michigan,
530 Church St., Ann Arbor, MI 48109, U.S.A.}
\email{kliechty@umich.edu}

\thanks{I would like to thank Gr\'{e}gory Schehr and Peter Forrester for bringing this problem to my attention, and Schehr for useful correspondence.  I would also like to thank Jinho Baik and Peter Miller for discussions, advice, and feedback.}

\date{\today}

\begin{abstract} 
We study the distribution of the maximal height of the outermost path in the model of $N$ nonintersecting Brownian motions on the half-line as $N\to \infty$, showing that it converges in the proper scaling to the Tracy-Widom distribution for the largest eigenvalue of the Gaussian orthogonal ensemble.  This is as expected from the viewpoint that the maximal height of the outermost path converges to the maximum of the $\textrm{Airy}_2$ process minus a parabola.  Our proof is based on Riemann-Hilbert analysis of a system of discrete orthogonal polynomials with a Gaussian weight in the double scaling limit as this system approaches saturation.  We consequently compute the asymptotics of the free energy and the reproducing kernel of the corresponding discrete orthogonal polynomial ensemble in the critical scaling in which the density of particles approaches saturation.  Both of these results can be viewed as dual to the case in which the mean density of eigenvalues in a random matrix model is vanishing at one point.
\end{abstract}

\maketitle

\section{Introduction}
\subsection{Nonintersecting Brownian motions on the half-line}
Consider a model of $N$ nonintersecting Brownian motions $\{b_j(t)\}_{j=1}^N$ which remain non-negative for $0<t<1$ and whose initial and terminal points are at zero.  That is,
\begin{equation}\label{in1}
\begin{aligned}
&b_1(0)=b_1(1)=b_2(0)=b_2(1)= \cdots = b_N(0)=b_N(1)=0\,, \\
0& \le b_1(t)<b_2(t)< \cdots < b_N(t) \quad \textrm{for} \quad 0<t<1.
\end{aligned}
\end{equation}
There are two standard ways to enforce the condition that the Brownian motions remain non-negative: an absorbing wall and a reflecting wall at zero.  The transition probability for a single Brownian motion with an absorbing wall at zero to pass from $y$ to $x$ over the time interval $t$ is given by
\begin{equation}\label{j4}
p_{abs}(t, x|y)=\frac{1}{\sqrt{2\pi t}} \left[\exp\left(-\frac{(y-x)^2}{2t}\right)-\exp\left(-\frac{(y+x)^2}{2t}\right)\right]\,,
\end{equation}
and the transition probability for a single Brownian motion with a reflecting wall at zero to pass from $y$ to $x$ over the time interval $t$ is given by
\begin{equation}\label{j4a}
p_{ref}(t, x|y)=\frac{1}{\sqrt{2\pi t}} \left[\exp\left(-\frac{(y-x)^2}{2t}\right)+\exp\left(-\frac{(y+x)^2}{2t}\right)\right].
\end{equation}
  For the positions at time $t\in (0,1)$ of the $N$ nonintersecting Brownian motions with an absorbing wall at zero, we will use the notation
\begin{equation}\label{j4b}
0 < b_1^{(BE)}(t)<b_2^{(BE)}(t)< \cdots < b_N^{(BE)}(t).
\end{equation}
The superscript $BE$ stands for Brownian excursion, which is the name for a Brownian motion with an absorbing wall which is conditioned to return to its starting point.  For the positions at time $t\in (0,1)$ of the $N$ nonintersecting Brownian motions with a reflecting wall at zero, we use the notation
\begin{equation}\label{j4c}
0 \le b_1^{(R)}(t)<b_2^{(R)}(t)< \cdots < b_N^{(R)}(t),
\end{equation}
where the superscript $R$ stands for reflecting.
The ensembles of nonintersecting Brownian motions can be derived from the transition probabilities (\ref{j4}) and (\ref{j4a}) using the Karlin-McGregor formula \cite{KM}.  Even though it seems like a degenerate condition to force all of the Brownian motions to begin and end at zero, it is possible to define these models of nonintersecting Brownian motions with such a condition by starting with a model for which the starting and ending points are all distinct and positive, and taking a limit as they go to zero.   
Let us now give a precise definition of the two models of nonintersecting Brownian motions in terms of their transition probabilities.  See \cite{KIK} for a derivation of these transition probabilities in the absorbing case.  The reflecting case is similar.
Introduce the notations
\begin{equation}\label{j1}
{\bf x}_k= (x_{k,1}, x_{k,2}, \dots, x_{k,N}),\quad {\bf x}_k^2=(x_{k,1}^2, x_{k,2}^2, \dots, x_{k,N}^2), 
\end{equation}
and let
\begin{equation}\label{j1a}
 \De({\bf x})=\prod_{1\leq j < k \leq N}(x_k-x_j),
 \end{equation}
 be the Vandermonde determinant.
Define the probability density functions
\begin{equation}\label{j2}
\begin{aligned}
p_0^{(BE)}(t,{\bf x}_1)&=\frac{\big(t(1-t)\big)^{-N^2-N/2} 2^{N/2}}{\pi^{N/2} \prod_{j=0}^{N-1} (2j+1)! } \big(\De({\bf x}_1^2)\big)^2 \left(\prod_{j=1}^N x_{1,j}^2\right) \exp\left(-\frac{1}{2t(1-t)} \sum_{j=1}^N x_{1,j}^2\right), \\
p_0^{(R)}(t,{\bf x}_1)&=\frac{\big(t(1-t)\big)^{-N^2+N/2} 2^{N/2}}{\pi^{N/2} \prod_{j=0}^{N-1} (2j)! } \big(\De({\bf x}_1^2)\big)^2  \exp\left(-\frac{1}{2t(1-t)} \sum_{j=1}^N x_{1,j}^2\right).
\end{aligned}
\end{equation}
and
\begin{equation}\label{j3}
\begin{aligned}
p^{(BE)}(t_1, {\bf x}_1;t_2,{\bf x}_2)=&\left(\frac{1-t_1}{1-t_2}\right)^{N^2+N/2} \left(\prod_{j=1}^N \frac{x_{2,j}}{x_{1,j}}\right) \frac{\De({\bf x}_2^2)}{\De({\bf x}_1^2)} \exp\left\{-\frac{1}{2}\sum_{j=1}^N \left(\frac{x_{2,j}}{1-t_2}-\frac{x_{1,j}}{1-t_1}\right)\right\} \\
&\qquad \times \det\bigg[p_{abs}(t_2-t_1, x_{2,j}|x_{1,k})\bigg]_{j,k=1}^N, \\
p^{(R)}(t_1, {\bf x}_1;t_2,{\bf x}_2)=&\left(\frac{1-t_1}{1-t_2}\right)^{N^2-N/2} \frac{\De({\bf x}_2^2)}{\De({\bf x}_1^2)} \exp\left\{-\frac{1}{2}\sum_{j=1}^N \left(\frac{x_{2,j}}{1-t_2}-\frac{x_{1,j}}{1-t_1}\right)\right\} \\
&\qquad \times \det\bigg[p_{ref}(t_2-t_1, x_{2,j}|x_{1,k})\bigg]_{j,k=1}^N. \\
\end{aligned}
\end{equation}
Let us adopt the convention that $b_j(t)$ with no superscript refers to either the model with the reflecting or absorbing wall at zero.
For any time $t\in (0,1)$, the ordered particles $(b_1(t), \dots ,b_N(t))$ must lie in the region in $\R^N$
\begin{equation}\label{j5}
{\bf W}_N=\{(b_1,\dots,b_n) : 0 \le b_1 < b_2 <\cdots < b_N\}.
\end{equation}
The meaning of the probability density functions above is the following.   For some sequence of times 
\begin{equation}\label{j6}
0<t_1<\cdots<t_K< 1,
\end{equation}
and some sequence of regions $\De_k \subset {\bf W}_N$, $k=1,\dots,K$ we have
\begin{equation}\label{j7}
\begin{aligned}
&\mathbb{P}\bigg[\big(b_1(t_k), b_2(t_k),\dots,b_N(t_k)\big) \in \De_k, \ k=1,\dots,K\bigg]= \\
&\qquad\qquad \int_{\De_1}\cdots \int_{\De_K} p_0(t_1,{\bf x}_1)\left(\prod_{k=1}^{K-1} p(t_k,{\bf x}_k;t_{k+1},{\bf x}_{k+1})\right) d{\bf x}_1 \cdots d{\bf x}_K,
\end{aligned}
\end{equation}
where  ${\bf x}_k$ is the vector of integration variables corresponding to the region $\De_k$, $p_0=p_0^{(BE)}$ (resp. $p_0^{(R)}$), and $p=p^{(BE)}$ (resp. $p^{(R)}$) for the model of nonintersecting Brownian motions with an absorbing (resp. reflecting) wall at zero.

From (\ref{j2}) it is immediate that for fixed $t \in (0,1)$, the particles $b_j^{(BE)}(t)$ are distributed as the eigenvalues of a random matrix from the Laguerre unitary ensemble (see e.g., \cite{For}).  As such, the largest particle at each fixed time, in the proper scaling limit, is distributed according to the distribution which describes the largest eigenvalue in the Gaussian unitary ensemble (GUE) of random matrices as the size of the matrix tends to infinity.  In this paper we study the distribution of the random variable
\begin{equation}\label{in2}
\max_{0<t<1} b_N(t),
\end{equation}
the maximal height of the outermost path for both the absorbing and reflecting case.  To obtain a formula for the distribution of this random variable, one can apply the Karlin-McGregor formula in the affine Weyl alcove of height $M$
\begin{equation}\label{in2a}
{\bf W}_N^M=\{(b_1,\dots,b_n) : 0 \le b_1 < b_2 <\cdots < b_N<M\}.
\end{equation}
In the absorbing wall case, the formula obtained is 
\begin{equation}\label{n1}
\begin{aligned}
\mathbb{P} \bigg(\max_{0<t<1} b_N^{(BE)}(t)<M\bigg) &= \frac{2^{-N/2}\pi^{2N^2+N/2}}{M^{N(2N+1)} N!\prod_{k=0}^{N-1}(2k+1)!} \\
& \quad \times \sum_{{\bf x} \in \Z^N} \big(\De({\bf x^2})\big)^2 \left(\prod_{j=1}^N x_j^2\right)\exp\left\{\frac{-\pi^2}{2M^2} \sum_{j=1}^N x_j^2\right\},
\end{aligned}
\end{equation}
where 
\begin{equation}\label{n1a}
{\bf x}= (x_1, x_2, \dots, x_N).
\end{equation}
Notice that this is the Hankel determinant 
\begin{equation}\label{n1aa}
\frac{2^{-N/2}\pi^{2N^2+N/2}}{M^{N(2N+1)}\prod_{k=0}^{N-1}(2k+1)!} \det\left[ \left.\left(\frac{\d^{j+k-2}}{\d \la^{j+k-2}} \sum_{x =-\infty}^\infty  x e^{\la x^2}\right)\right|_{\la=-\pi^2/2M^2}  \right]_{j,k=1}^N.
\end{equation}
This formula first appeared in the paper \cite{SMCR}, in which the authors use a path integral technique to derive it.  A derivation using the Karlin-McGregor formula appeared soon after in \cite{KIK}.   See also \cite{Fe} in which an equivalent formula is derived from lattice paths.  
In the case that there is a reflecting wall at zero, the formula is
\begin{equation}\label{n1b}
\mathbb{P} \bigg(\max_{0<t<1} b_N^{(R)}(t)<M\bigg) =\frac{2^{-N/2}\pi^{2N^2-3N/2}}{M^{N(2N-1)} N!\prod_{k=0}^{N-1}(2k)!}  \sum_{{\bf x} \in \{\Z-1/2\}^N} \big(\De({\bf x^2})\big)^2 \exp\left\{\frac{-\pi^2}{2M^2} \sum_{j=1}^N x_j^2\right\}.
\end{equation}
To our knowledge, the formula (\ref{n1b}) has not appeared in the literature before, but it can be derived in a manner similar to those used in \cite{SMCR} and \cite{KIK} to derive (\ref{n1}).

In the case of an absorbing wall, this model was first introduced in the papers \cite{KT} and \cite{Gi}, and is often called the model of nonintersecting Brownian excursions, or noncolliding Bessel bridges.  It is also sometimes referred to as ``watermelons with a wall," although this sometimes refers to the discrete time and space (simple random walk) version as well \cite{Gi}.   See \cite{Gi} and \cite{KTNK} for a derivation of this model as a scaling limit of an ensemble of simple random walks conditioned not to intersect and to stay positive.

  Our analysis of (\ref{n1}) and (\ref{n1b}) is based on analysis of the {\it discrete Gaussian orthogonal polynomials} $\{P_k^{(\al)}(x)\}_{k=0}^\infty$  and their normalizing constants $\{h_k^{(\al)}\}_{k=0}^\infty$ defined via the orthogonality condition
\begin{equation}\label{n5}
\sum_{x\in \{\Z-\al\}} P_n^{(\al)}(x)P_m^{(\al)}(x) w(x) =h_n^{(\al)} \de_{mn}\,, \quad w(x)= \exp\left\{\frac{-\pi^2}{2M^2}x^2\right\}, \quad P_n^{(\al)}(x)=x^n+\cdots.
\end{equation}
A routine calculation (see \cite{De}) shows that (\ref{n1}) and (\ref{n1b}) can be written as
\begin{equation}\label{n4}
\begin{aligned}
\mathbb{P} \bigg(\max_{0<t<1} b_N^{(BE)}(t)<M\bigg) &= \frac{2^{-N/2}\pi^{2N^2+N/2}}{M^{N(2N+1)} \prod_{k=0}^{N-1}(2k+1)!} \prod_{k=0}^{N-1} h_{2k+1}^{(0)}\,, \\
\mathbb{P} \bigg(\max_{0<t<1} b_N^{(R)}(t)<M\bigg) &=\frac{2^{-N/2}\pi^{2N^2-3N/2}}{M^{N(2N-1)} \prod_{k=0}^{N-1}(2k)!}  \prod_{k=0}^{N-1} h_{2k}^{(1/2)}.
\end{aligned}
\end{equation}

In a recent paper of Forrester, Majumdar, and Schehr \cite{FMS}, an analogy beween nonintersecting Brownian excursions and Yang-Mills theory on the sphere is made, and the authors use some non-rigorous methods from gauge theory (\cite{CNS},\cite{DK})  to deduce that the maximal height of the outermost path in this ensemble is, in the proper scaling limit, distributed as the largest eigenvalue in the Gaussian orthogonal ensemble (GOE) of random matrices.  A main result of this paper is a rigorous verification of this fact.  In order to state this theorem, let us review the Tracy-Widom distributions which describe the location of the largest eigenvalue in the classical random matrix ensembles.  These distributions may be described in terms of the Hastings-McLeod solution to the Painlev\'{e} II equation.  The homogeneous Painlev\'{e} II equation is the second order nonlinear ODE 
\begin{equation}\label{in3}
q''(s)=sq(s)+2q(s)^3\,.
\end{equation}
The Hastings-McLeod solution to this equation \cite{HM} is characterized by its behavior at positive infinity.  In particular, it is the solution satisfying 
\begin{equation}\label{in4}
q(s)=\Ai(s)(1+o(1))=\frac{e^{-\frac{2}{3}s^{3/2}}}{2\sqrt{\pi} s^{1/4}} \bigg(1+o(1)\bigg) \quad \textrm{as} \quad s\to +\infty,
\end{equation}
where $\Ai$ is the Airy function.   The distribution functions $\fcal_1$ and $\fcal_2$ are defined as
\begin{equation}\label{in5}
\fcal_1(x)=F(x)E(x)\,, \quad \fcal_2(x)=F(x)^2\,,
\end{equation}
where
\begin{equation}\label{in6}
E(x)=\exp\left(-\frac{1}{2}\int_x^\infty q(s) ds\right)\,, \quad F(x)=\exp\left(-\frac{1}{2}\int_x^\infty R(s) ds\right)\,, \quad R(x)=\int_x^\infty q(s)^2 ds,
\end{equation}
and $q(s)$ is the Hastings-McLeod solution to the Painlev\'{e} II equation.  The function $\fcal_1$ describes the distribution of the largest (or smallest) rescaled eigenvalue in GOE, while $\fcal_2$ describes the distribution of the largest (or smallest) rescaled eigenvalue in GUE (see \cite{TW1}, \cite{TW2}, \cite{TW3}).
We now state our main theorem.
\begin{theo}\label{maxheight} (Distribution of the maximal height of the outermost particle)
Consider either of the models of nonintersecting Brownian motions described in (\ref{j2}) and (\ref{j3}).  The maximal height of the outermost particle has the limiting distribution
\begin{equation}\label{main21}
\lim_{N\to \infty} \mathbb{P} \left[2^{11/6} N^{1/6} \left(\max_{0<t<1} b_N(t)-\sqrt{2N}\right)< k \right]=\fcal_1(k),
\end{equation}
where $\fcal_1$, defined in (\ref{in5}) and (\ref{in6}), is the limiting distribution function for the location of the largest eigenvalue in the Gaussian orthogonal ensemble of random matrices, and $b_N(t)$ is either $b_N^{(BE)}(t)$ or $b_N^{(R)}(t)$.
\end{theo}

This theorem is widely expected from the point of view that the distribution of the uppermost curve in the model of $N$ nonintersecting Brownian bridges should converge (after rescaling and recentering) to the $\textrm{Airy}_2$ process, which was first introduced in \cite{PS}.   In the case of an absorbing wall, the framework to prove this convergence at the level of finite dimensional distributions was given by Tracy and Widom in \cite{TW}, although they stopped just short of stating it as a theorem (their main interest in that paper was the asymptotics of the bottom curve).  
It is known that the maximum of the $\textrm{Airy}_2$ process over a continuum of times is given by the Tracy-Widom GOE distribution.  Such a result was first proved by Johansson \cite{Jo} by first proving a functional limit theorem for the convergence of the polynuclear growth (PNG) model to the $\textrm{Airy}_2$ process and using connections between PNG and the longest increasing subsequence of a random permutation found by Baik and Rains \cite{BR}.  A more direct proof was recently given by Corwin, Quastel and Remenik \cite{CQR}.  See also \cite{MQR}.  Thus Theorem \ref{maxheight} could be proved by establishing the functional convergence of the top curve to the $\textrm{Airy}_2$ process.  In fact, for the absorbing boundary case, given the finite dimensional convergence implied by  \cite{TW} the functional convergence follows from a result of Corwin and Hammond \cite{CH}, who showed that finite dimensional convergence of a line ensemble implies functional convergence given some fairly mild local condition.  In this sense, at least in the absorbing case, Theorem \ref{maxheight} is not new, but here we give an alternative direct proof by analyzing the formula (\ref{n1}) asymptotically.  Moreover, the analysis is uniform for both the absorbing and reflecting boundaries.  The rigorous result for the reflecting boundary case does seem to be new.  Our proof is based on the asymptotic evaluation of the formulas (\ref{n4}) by Riemann-Hilbert methods.

Let us note here that in the paper \cite{FMS} the authors give expressions similar to (\ref{n1}) for the normalized reunion probabilities for nonintersecting Brownian motions with periodic boundary conditions and with reflecting boundary conditions, which are shown to correspond to partition functions of 2-d Yang-Mills theory on the sphere with the gauge groups U($N$) and SO($2N$), respectively.   These expressions do not have a probabilistic interpretation, but can also be expressed in terms of discrete Gaussian orthogonal polynomials and their asymptotic evaluation is a straightforward application of Theorems \ref{maxheight} and \ref{dopefreeenergy}.

\subsection{Discrete Gaussian orthogonal polynomials}
For asymptotic analysis, it is convenient to use a rescaling of the polynomials (\ref{n5}).  Consider the infinite regular lattice of mesh $1/n$,
\begin{equation}\label{dope1}
L_{n,\al}=\left\{x_k=\frac{k-\al}{n},\quad k \in \Z\right\}, \quad \al \in \left[-\frac{1}{2}, \frac{1}{2}\right].
\end{equation}
and the polynomials orthogonal with respect to a discrete Gaussian weight on this lattice.  More specifically, consider the system of monic polynomials $\{P_{n,j}^{(\al)}(x)\}_{j=0}^\infty$ and the normalizing constants $\{h_{n,j}^{(\al)}\}_{j=0}^\infty$ satisfying the orthogonality condition
\begin{equation}\label{dope2}
\frac{1}{n}\sum_{x \in L_{n,\al}} P_{n,j}^{(\al)}(x)P_{n,k}^{(\al)}(x)e^{-n\frac{\pi^2 a}{2}x^2}=h_{n,k}^{(\al)} \de_{jk}.
\end{equation}
As usual, $P_{n,k}^{(\al)}(x)$ is a polynomial of degree $k$, and $a>0$ is a positive parameter.  As the mesh of the lattice goes to zero these polynomials converge to the (monic and rescaled) Hermite polynomials.  The polynomials $P_{n,k}^{(\al)}(x)$ and the normalizing constants $h_{n,k}^{(\al)}$ depend on the parameter $a$.  To highlight that dependence, let us write
\begin{equation}\label{dope3}
h_{n,k}^{(\al)} \equiv h_{n,k}^{(\al)}(a).
\end{equation}
The relation to the polynomials (\ref{n5}) is
\begin{equation}\label{rop3}
P_k^{(\al)}(nx)=n^kP_{n,k}^{(\al)}(x)\,, \qquad h_{k}^{(\al)}=n^{2k+1} h_{n,k}^{(\al)}\left(a\right)\,, \quad M=\sqrt{\frac{n}{a}}\,.
\end{equation}

The main distinguishing feature between the asymptotic analysis of discrete orthogonal polynomials and that of continuous ones is the phenomenon of saturation.  If we denote by $\mu_n$ the normalized counting measure on the zero set of the polynomials $P_{n,n}^{(\al)}(x)$, it is known that, as $n\to \infty$, $\mu_n$ converges to a probability measure with finite support and piecewise smooth density, known as the {\it equilibrium measure}.  It is a general fact that for any system of polynomials orthogonal with respect to a measure which lies on a discrete subset of $\R$, call it $D$, all zeroes of the polynomials are real and there can be no more than one zero between two consecutive nodes of $D$.  This leads to an upper constraint on the distribution of zeroes.  For polynomials orthogonal with respect to a continuous weight, there is no such constraint.  

If the discrete orthogonal polynomials are such that the equilibrium measure does not approach this upper constraint, then their asymptotic properties match those of a corresponding continuous system.  If the upper constraint is active, then they do not, see \cite{BKMM}, \cite{BL}.  In the case of the discrete Gaussian orthogonal polynomials described in (\ref{dope2}), the mesh of the lattice is $1/n$ and thus the the upper constraint on the equilibrium measure is that it should have a density which is no greater than 1.  In the case that the parameter $a$ is greater than 1, this upper constraint is realized, meaning that there is an interval on which the density of the equilibrium measure is identically 1, see Figure \ref{sat}.  
\begin{center}
\begin{figure}
\begin{center}
\scalebox{0.39}{\includegraphics{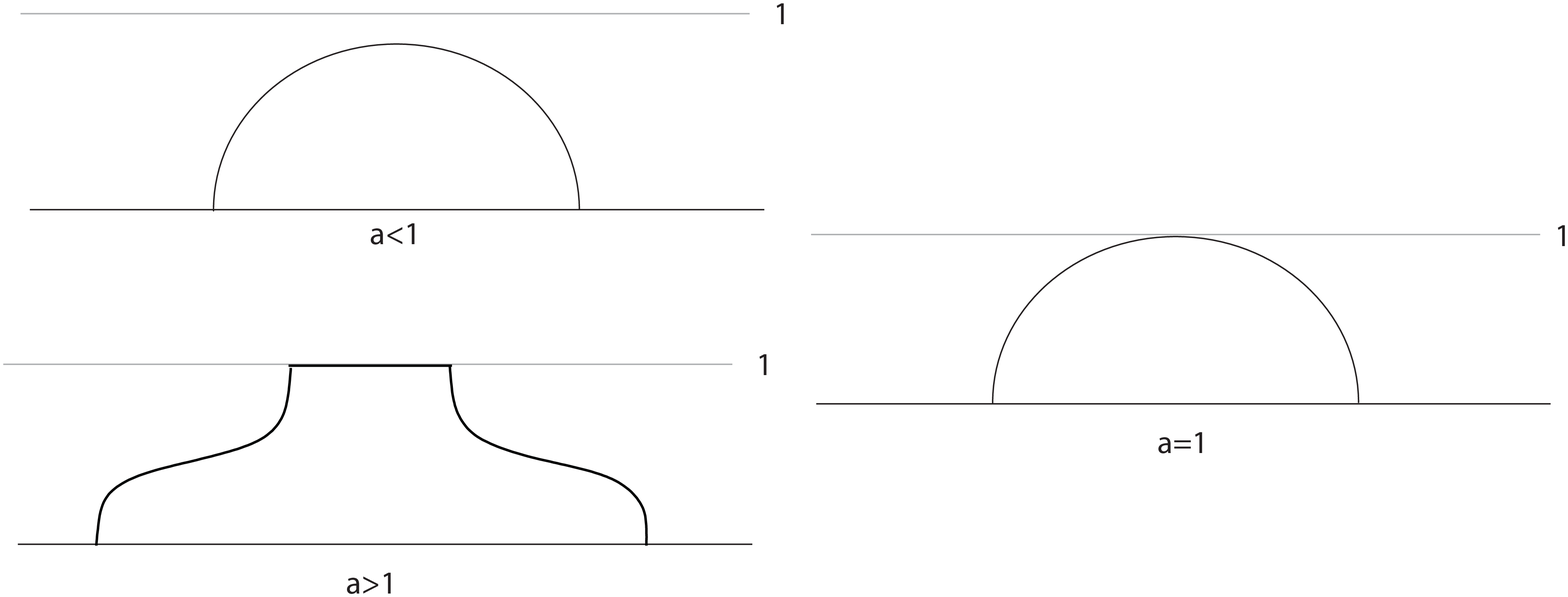}}
\caption{The equilibrium measure density for discrete Gaussian orthogonal polynomials in the subcritical case $a<1$, the supercritical case $a>1$, and the critical case $a=1$.}
\label{sat}
\end{center}
\end{figure}
\end{center}

In \cite{BKMM} and \cite{BL} the authors present the asymptotic properties of very general classes of discrete orthogonal polynomials on the real line assuming some ``regularity" condition on the equilibrium measure. The proof of Theorem \ref{maxheight} requires that we explore the critical case in which the upper constraint is approached in a double scaling limit, which is not considered in \cite{BKMM} or \cite{BL}.  For the continuous weight case, a similar double scaling limit was studied in the seminal paper of Baik, Deift, and Johannson \cite{BDJ} in the context of the longest increasing subsequence of a random permutation.  In the context of random matrix theory, this type of double scaling limit appears  when the limiting mean density of eigenvalues in a random matrix model vanishes at one point, see \cite{BI1}, \cite{CK}.  In the present paper we adapt the analysis to a discrete weight where the lower constraint is replaced by the upper constraint.  A similar double scaling limit was considered recently by Baik and Jenkins \cite{BJ}  for a system of discrete orthogonal polynomials on the circle when the upper constraint is about to be active.

The orthogonal polynomials (\ref{dope3}) satisfy the recurrence relation (see e.g., \cite{Sze})
\begin{equation}\label{dope4}
xP_{n,k}^{(\al)}(x)=P_{n,k+1}^{(\al)}(x)+A_{n,k}^{(\al)} P_{n,k}^{(\al)}(x)+B_{n,k}^{(\al)}P_{n,k-1}^{(\al)}(x), \quad B_{n,k}^{(\al)} \equiv B_{n,k}^{(\al)}(a)=\frac{h_{n,k}^{(\al)}(a)}{h_{n,k-1}^{(\al)}(a)}.
\end{equation}
In the case that $\al=0$ or $\al=1/2$, the lattice is symmetric about zero and $A_{n,k}^{(\al)}$ is zero.  The polynomials $P_{n,j}^{(\al)}$ of course depend on the parameter $a$, and if we show that dependence by writing $P_{n,j}^{(\al)}( x ; a)$, then we have the lattice rescaling relations
\begin{equation}\label{dope4a}
P_{n,j}^{(\al)}(\xi_{\pm} x ; a)=\xi_{\pm}^j P_{n\pm1,j}^{(\al)}(x; a\xi_{\pm})\,, \quad h_{n,j}^{(\al)}(a)=\xi_\pm^{2j+1} h_{n\pm1,j}^{(\al)}(a\xi_\pm)\,, \quad A_{n,k}^{(\al)}(a)=\xi_\pm A_{n\pm1,k}^{(\al)}(a\xi_\pm)\,,
\end{equation}
where
\begin{equation}\label{dope4b}
\xi_{\pm}= 1\pm \frac{1}{n}.
\end{equation}

A basic physical model described by these orthogonal polynomials is the {\it discrete orthogonal polynomial ensemble} with Gaussian weight, which is a discrete version of GUE.  This ensemble is described as the probability distribution on $n$-tuples of points $\la=(\la_1,\la_2,\dots,\la_n)\in (L_{n,\al})^n$
\begin{equation}\label{dope5}
\mathbb{P}\big(\textrm{there are particles at each of the points } \la_1,\dots,\la_n\big)=\left(Z_n^{(DOPE)}\right)^{-1}\prod_{k>j}(\la_j-\la_k)^2\prod_{j=1}^n \frac{e^{-nV(\la_j)}}{n},
\end{equation}
where
\begin{equation}\label{dope6}
Z_n^{(DOPE)} =\sum_{\la\in (L_{n,\al})^n}\prod_{1\le j<k\le n}(\la_k-\la_j)^2\prod_{j=1}^n \frac{e^{-nV(\la_j)}}{n}=n!\prod_{k=0}^{n-1} h_{n,k}^{(\al)},
\end{equation}
and
\begin{equation}\label{dope7}
V(x)=\frac{\pi^2 a}{2} x^2.
\end{equation}
As proved in Appendix \ref{defOP}, the partition function $Z_n^{(DOPE)}$ satisfies the deformation equation
\begin{equation}\label{dope9}
\frac{\d^2}{\d a^2} \log Z_n^{(DOPE)}=\left(\frac{n\pi^2}{2}\right)^2 B_{n,n}^{(\al)}\bigg(B_{n,n-1}^{(\al)}+B_{n,n+1}^{(\al)}+\big(A_{n,n}^{(\al)}+A_{n,n-1}^{(\al)}\big)^2\bigg).
\end{equation}
The deformation (\ref{dope9}) is one of the isospectral flows in a general system known as the {\it Toda lattice hierarchy}, see \cite{DLT}, \cite{DNT} and references therein.  For a derivation of the Toda lattice from the Riemann-Hilbert problem for orthogonal polynomials, see \cite{BKMM}, and for a broad description of differential equations related to orthogonal polynomials, see \cite{Eyn}.  Let us note that there also exist similar deformations with respect to other parameters in the weight which yield closed form expressions for the first logarithmic derivative of the partition function, see \cite{DIK}, \cite{IK}, \cite{Kra}.

If we define the free energy as
\begin{equation}\label{dope10}
F_n^{(DOPE)}=-\frac{1}{n^2}\log Z_n^{(DOPE)},
\end{equation}
then (\ref{dope9}) reads as
\begin{equation}\label{dope11}
\begin{aligned}
\frac{\d^2}{\d a^2} F_n^{(DOPE)}&=-\left(\frac{\pi^2}{2}\right)^2 B_{n,n}^{(\al)}\bigg(B_{n,n-1}^{(\al)}+B_{n,n+1}^{(\al)}+\big(A_{n,n}^{(\al)}(a)+A_{n,n-1}^{(\al)}(a)\big)^2\bigg) \\
&=-\left(\frac{\pi^2}{2}\right)^2 \left(\frac{h_{n,n}^{(\al)}(a)}{h_{n,n-2}^{(\al)}(a)}+\frac{h_{n,n+1}^{(\al)}(a)}{h_{n,n-1}^{(\al)}(a)}+\frac{h_{n,n}^{(\al)}(a)}{h_{n,n-1}^{(\al)}(a)}\big(A_{n,n}^{(\al)}(a)+A_{n,n-1}^{(\al)}(a)\big)^2\right).
\end{aligned}
\end{equation}
In light of the deformation equations (\ref{dope9}) and (\ref{dope11}), let us write
\begin{equation}\label{dope11a}
Z_n^{(DOPE)} \equiv Z_n^{(DOPE)}(a)\,, \qquad F_n^{(DOPE)} \equiv F_n^{(DOPE)}(a).
\end{equation}

All correlation functions for this ensemble can be written in terms of a reproducing kernel which is defined in terms of orthogonal polynomials.  
Introduce the $\psi$-functions
\begin{equation}\label{dope12}
\psi_{n,k}(x)=\frac{1}{\sqrt{h_{n,k}^{(\al)}}}P_{n,k}^{(\al)}(x)\frac{e^{-nV(x)/2}}{\sqrt{n}}\,,
\end{equation}
and the Christoffel-Darboux kernel
\begin{equation}\label{dope13}
K_n(x,y)=\sum_{k=0}^{n-1}\psi_{n,k}(x)\psi_{n,k}(y).
\end{equation}
Then the $m$-point correlation function $R_{m,n}(\la_1,\dots \la_n)$ is defined by the formula
\begin{equation}\label{dope14}
R_{m,n}(\la_1,\dots \la_m)=\det\big(K_n(\la_k,\la_l)\big)_{k,l=1}^m.
\end{equation}

The proof of Theorem \ref{maxheight} requires an asymptotic formula for $h_{n,k}^{(\al)}(a)$ in the scaling limit $a\to 1$ as $n\to \infty$.  In order to state that expansion, let us first fix some notations.  Let the parameter $s$ be defined in terms of $a$ as
\begin{equation}\label{main1}
s\equiv s(a;n)=-\left[(3\pi n)\left(z_1-\int_0^{z_1} \rho(\eta) d\eta \right)\right]^{2/3},
\end{equation}
where
\begin{equation}\label{main2}
z_1=\frac{2}{\pi a} \sqrt{a-1}\,, \quad \rho(\eta)=\frac{\pi a}{2}\sqrt{\frac{4}{\pi^2 a}-\eta^2}\,.
\end{equation}
One may check that as $a\to 1$, 
\begin{equation}\label{main3}
s=2^{2/3} n^{2/3} \left((1-a)+\frac{4}{5}(1-a)^2+\frac{122}{175}(1-a)^3+O\big((1-a)^4\big)\right),
\end{equation}
and thus if 
\begin{equation}\label{main4}
a=1- x n^{-\de} +O(n^{-\ep})\,, \quad 0<\de<\ep,
\end{equation}
then
\begin{equation}\label{main5}
\lim_{n\to\infty} s=\left\{ 
\begin{aligned}
(\textrm{sgn}\, x) \infty \quad \textrm{if} \quad 0<\de<2/3\\
2^{2/3} x \quad \textrm{if} \quad \de=2/3 \\
0 \quad \textrm{if} \quad 2/3<\de.
\end{aligned}
\right.
\end{equation}
Let $q(s)$ be the Hastings-McLeod solution to the Painlev\'{e} II equation, and $R(s)$ be as defined in (\ref{in6}).  We then have the following proposition.

\begin{prop}\label{asymptoticsh} (Asymptotics in discrete Gaussian orthogonal polynomials)
Consider the orthogonal polynomials defined in (\ref{dope2}) such that $(1-a)n^{2/3}$ remains bounded as $n\to \infty$.  Let $s$ be defined in terms of $a$ as in (\ref{main1}).  The normalizing constants for the orthogonal polynomials (\ref{dope2}) satisfy, as $n\to \infty$,
\begin{equation}\label{main6}
\begin{aligned}
h_{n,n}^{(\al)}(a)&=\frac{2}{\sqrt{a}}\left(\frac{1}{\pi^2 a e}\right)^n\left(1-\frac{2^{2/3}}{n^{1/3}}T_n(s)+\frac{2^{1/3}}{n^{2/3}}U_n(s)+O(n^{-1})\right)  \\
h_{n,n-1}^{(\al)}(a)^{-1}&=\frac{1}{2\sqrt{a}\pi^2}\left({\pi^2 a e}\right)^n\left(1+ \frac{2^{2/3}}{n^{1/3}}T_{n-1}(s)+\frac{2^{1/3}}{n^{2/3}}U_{n-1}(s)+O(n^{-1})\right) .
\end{aligned}
\end{equation}
where
\begin{equation}\label{main7}
\begin{aligned}
T_n(s)&=R(s)-(-1)^n \cos(2\pi\al)q(s)\,, \\
U_n(s)&=R^2(s)-(-1)^n\cos(2\pi\al)\big(q'(s)+2q(s)R(s)\big)-q^2(s)\sin^2(2\pi\al).
\end{aligned}
\end{equation}
The recurrence coefficients $A_{n,n-1}^{(\al)}$ satisfy
\begin{equation}\label{main7a}
A_{n,n-1}^{(\al)}(a)=\frac{(-1)^n 2^{4/3} \sin(2\pi \al) }{\pi n^{1/3}}\left(2^{1/3} q(s)+\frac{q'(s)}{n^{1/3}}+O(n^{-2/3})\right).
\end{equation}
\end{prop}

Let us also note a formula for ratios of normalizing constants which will be useful in the proofs of Theorems \ref{maxheight} and \ref{dopefreeenergy}.  Introduce the notations
\begin{equation}\label{ins1}
s_+=s(a\xi_+;n+1)\,, \qquad s_-=s(a\xi_-;n-1)\,, \qquad s=s(a;n)\,,
\end{equation}
where $s(a;n)$ is defined in (\ref{main1}), and $\xi_{\pm}$ is as defined in (\ref{dope4b}).  Assume that $a=1-xn^{-2/3}$ for $x\in \R$.  A direct application of Proposition \ref{asymptoticsh} gives
\begin{equation}\label{ins2}
\begin{aligned}
\frac{h_{n,n}^{(\al)}(a)}{h_{n-1,n-2}^{(\al)}(\xi_-a)}&=\frac{1}{\pi^4 a^2 e^2} \left(1+\frac{2^{2/3}}{n^{1/3}} \big(T_n(s_-)-T_n(s)\big)\right. \\
&\left. \hspace{4cm} + \frac{2^{1/3}}{n^{2/3}} \big( U_n(s_-)+U_n(s)-2T_n(s)T_n(s_-)\big) +O(n^{-1})\right)\,, \\
\frac{h_{n+1,n+1}^{(\al)}(\xi_+a)}{h_{n,n-1}^{(\al)}(a)}&=\frac{1}{\pi^4 a^2 e^2} \left(1+\frac{2^{2/3}}{n^{1/3}} \big(T_{n+1}(s)-T_{n+1}(s_+)\big) \right. \\
&\left. \hspace{3cm} + \frac{2^{1/3}}{n^{2/3}} \big( U_{n+1}(s_+)+U_{n+1}(s)-2T_{n+1}(s)T_{n+1}(s_+)\big) +O(n^{-1})\right)\,.
\end{aligned}
\end{equation}
From (\ref{main3}) we see that
\begin{equation}\label{ins3}
s-s_\pm = \pm \frac{2^{2/3}}{n^{1/3}}+O(n^{-1})\,,
\end{equation} 
and thus we can rewrite (\ref{ins2}) as
\begin{equation}\label{ins4}
\begin{aligned}
\frac{h_{n,n}^{(\al)}(a)}{h_{n-1,n-2}^{(\al)}(\xi_-a)}&=\frac{1}{\pi^4 a^2 e^2} \left(1+ \frac{2^{4/3}}{n^{2/3}} \big( T_n'(s)+U_n(s)-T_n(s)^2\big) +O(n^{-1})\right)\,, \\
\frac{h_{n+1,n+1}^{(\al)}(\xi_+a)}{h_{n,n-1}^{(\al)}(a)}&=\frac{1}{\pi^4 a^2 e^2} \left(1+ \frac{2^{4/3}}{n^{2/3}} \big(T_{n+1}'(s) +U_{n+1}(s)-T_{n+1}(s)^2\big) +O(n^{-1})\right)\,.
\end{aligned}
\end{equation}
In fact, it easy to see, using $R'(s)=-q^2(s)$, that $U_n(s)-T_n(s)^2=T_n'(s)$, and thus we have
\begin{equation}\label{ins5}
\begin{aligned}
\frac{h_{n,n}^{(\al)}(a)}{h_{n-1,n-2}^{(\al)}(\xi_-a)}&=\frac{1}{\pi^4 a^2 e^2} \left(1+ \frac{2^{7/3}T_n'(s)}{n^{2/3}}  +O(n^{-1})\right)\,, \\
\frac{h_{n+1,n+1}^{(\al)}(\xi_+a)}{h_{n,n-1}^{(\al)}(a)}&=\frac{1}{\pi^4 a^2 e^2} \left(1+ \frac{2^{7/3}T_{n+1}'(s)}{n^{2/3}} +O(n^{-1})\right)\,.
\end{aligned}
\end{equation}

A study of the asymptotic properties of the discrete Gaussian orthogonal polynomials in the critical scaling above naturally leads to asymptotic results for the discrete Gaussian orthogonal polynomial ensemble (\ref{dope5}) in the critical scaling such that the distribution of particles is approaching saturation.  The theorems below give such results, and we note that they are nearly identical to the results in \cite{BI1}, \cite{BI2}, and \cite{CK}, which concern a random matrix model for which the distribution of eigenvalues is vanishing at a single point.

We first compare the free energy in the discrete Gaussian orthogonal polynomial ensemble with that of the Gaussian unitary ensemble.  The free energy of GUE is defined as
\begin{equation}\label{main18}
F_n^{(GUE)}=-\frac{1}{n^2}\log Z_n^{(GUE)}\,, \quad Z_n^{(GUE)}=\int \cdots \int_{\R^n} \prod_{1\le j < k \le n} (x_k-x_j)^2 \exp\left(-\sum_{j=1}^n x_j^2\right)dx_1\cdots dx_n.
\end{equation}

\begin{theo}\label{dopefreeenergy} (Free energy in the discrete Gaussian orthogonal polynomial ensemble)
Let the parameter $a$ be such that $(1-a)n^{2/3}$ remains bounded as $n \to \infty$.
The free energy of the discrete orthogonal polynomial ensemble, $F_n^{(DOPE)}(a)$, defined in (\ref{dope6}) and (\ref{dope10}), satisfies as $n\to \infty$ 
\begin{equation}\label{main19}
F_n^{(DOPE)}(a)-F_n^{(GUE)}=\frac{\log(a)}{2}-\frac{1}{2}\log\left(\frac{2}{n\pi^2}\right)-\frac{1}{n^2}\log\bigg[\fcal_2\big(2^{2/3}n^{2/3}(1-a)\big)\bigg]+O(n^{-\de}),
\end{equation}
where $\fcal_2(x)$ is the Tracy-Widom distribution function associated with the largest eigenvalue of GUE, and $2<\de<7/3$.
\end{theo}

We would now like to describe the Christoffel-Darboux kernel (\ref{dope13}) close to the origin as $n\to\infty$.  In order to state the theorem, we need to fix some notation.  Let $\Phi^1(\z;s)$ and $\Phi^2(\z;s)$ be defined via the system of differential equations
\begin{equation}\label{main8}
\begin{aligned}
\frac{\partial}{\partial \z} \begin{pmatrix} \Phi^1(\z;s) \\ \Phi^2(\z;s) \end{pmatrix} &=\begin{pmatrix} 4\z q & 4\z^2 +s +2q^2+2r \\ -4\z^2-s-2q^2+2r & -4\z q \end{pmatrix}\begin{pmatrix} \Phi^1(\z;s) \\ \Phi^2(\z;s) \end{pmatrix} \\
\frac{\partial}{\partial s} \begin{pmatrix} \Phi^1(\z;s) \\ \Phi^2(\z;s) \end{pmatrix} &=\begin{pmatrix}  q & \z \\ -\z & -q \end{pmatrix}\begin{pmatrix} \Phi^1(\z;s) \\ \Phi^2(\z;s) \end{pmatrix}\,, \\
\end{aligned}
\end{equation}
satisfying the properties that $\Phi^1(\z;s)$ and $\Phi^2(\z;s)$ are real for real $\z$ and $s$, 
\begin{equation}\label{main9}
\Phi^1(-\z;s)=\Phi^1(\z;s), \quad \Phi^2(-\z;s)=-\Phi^2(\z;s),
\end{equation}
and have the real asymptotics
\begin{equation}\label{main10}
\Phi^1(\z;s)=\cos\left(\frac{4}{3}\z^3+s\z\right)+O(\z^{-1}), \quad \Phi^2(\z;s)=-\sin\left(\frac{4}{3}\z^3+s\z\right)+O(\z^{-1})\,,
\end{equation}
as $\z \to \pm \infty$.  These are the so-called psi-functions associated with the Painlev\'{e} II equation. We then define the functions
\begin{equation}\label{main11}
\Phi_1(\z;s)=\Phi^1(\z;s)+i\Phi^2(\z;s), \quad \Phi_2(\z;s)=\Phi^1(\z;s)-i\Phi^2(\z;s),
\end{equation}
and the critical kernel
\begin{equation}\label{main12}
K^{crit}(u,v;s)=\frac{-\Phi_1(u;s)\Phi_2(v;s)+\Phi_2(u;s)\Phi_1(v;s)}{2\pi i(u-v)}.
\end{equation}
Let us also give two other expressions for $K^{crit}$ which may be useful for analysis (see e.g., \cite{CK}),
\begin{equation}\label{main13}
\begin{aligned}
K^{crit}(u,v;s)&=\frac{\Phi^1(u;s)\Phi^2(v;s)-\Phi^2(u;s)\Phi^1(v;s)}{\pi(u-v)} \\
&=\frac{1}{\pi}\int_{-\infty}^s\bigg[\Phi^1(u;\xi)\Phi^1(v;\xi)+\Phi^2(u;\xi)\Phi^2(v;\xi)\bigg]d\xi.
\end{aligned}
\end{equation}
As shown in \cite{BI1} and \cite{CK}, this is the limiting correlation kernel of a random matrix model in the case that the limiting distribution of eigenvalues vanishes at a single point.
We have the following expression for the kernel (\ref{dope13}) as $n\to \infty$.

\begin{theo}\label{kernel}(Limiting Christoffel-Darboux kernel near the point of saturation)
Let $K_n(x,y)$ be the Christoffel-Darboux kernel defined in (\ref{dope13}) for the discrete orthogonal polynomial ensemble (\ref{dope5}) with potential $V(x)$ given by (\ref{dope7}).  Consider the scaling
\begin{equation}\label{main14}
a = 1-Ln^{-2/3}, \quad x=\frac{k_n-\al}{n}\sim \frac{u}{cn^{1/3}}, \quad y=\frac{m_n-\al}{n}\sim \frac{v}{cn^{1/3}}, \quad c=\pi2^{-5/3},
\end{equation}
where $k_n$ and $m_n$ are integers.  Then, for $u \neq v$, 
\begin{equation}\label{main15}
\lim_{n\to \infty} (-1)^{k_n+m_n+1} \frac{n^{2/3}}{c} K_n(x,y)=K^{crit}(u,v;s_\infty),
\end{equation}
where
\begin{equation}\label{main16}
s_\infty=2^{2/3}L.
\end{equation}
The diagonal terms satisfy
\begin{equation}\label{main17}
\lim_{n\to \infty}\frac{n^{2/3}}{c}\bigg(1-K_n(x,x)\bigg)=K^{crit}(u,u;s_\infty),
\end{equation}
where $K^{crit}(u,u;s_\infty)$ is obtained from (\ref{main12}) using L'Hospital's rule or directly from (\ref{main13}).
\end{theo}
Theorem \ref{kernel} may have some application to nonintersecting Brownian excursions.  In a recent paper of Rambeau and Schehr \cite{RS}, the authors derive a formula for the joint distribution of the maximal height of the outermost path and the time at which it occurs.  Their formula can be written in terms of the Christoffel-Darboux kernel (\ref{dope13}) \cite{Sch2}, and thus Theorem \ref{kernel} may be of use in the asymptotic analysis of this joint distribution.  In fact, very recently Schehr gave a limiting formula for this joint distribution in the critical scaling which involves the Painlev\'{e} II psi-function \cite{Sch2}.  The argument of \cite{Sch2} is based on a differential Ansatz, and it would be interesting to see if one could give a rigorous verification of that result using Theorem \ref{kernel}.

The rest of this paper is organized as follows.  In section 2, we derive an integral formula for the distribution of the random variable $\max_{0<t<1} b_N(t)$ in terms of discrete Gaussian orthogonal polynomials and use Proposition \ref{asymptoticsh} to evaluate it in the large $N$ limit, which proves Theorem \ref{maxheight}.  In section 3, we prove Theorem \ref{dopefreeenergy} in a similar way.  In section 4, we present the steepest descent analysis of a Riemann-Hilbert problem for the discrete Gaussian orthogonal polynomals by the method of Deift and Zhou \cite{DZ2}, and in section 5 we explicity compute the first two error terms of this analysis in the critical scaling limit.  Finally, in section 6 we use the results of sections 4 and 5 to prove Proposition \ref{asymptoticsh} and Theorem \ref{kernel}.

\medskip

\section{Proof of Theorem \ref{maxheight}}
\subsection{Integral formula for the distribution of the maximal height of $b_N(t)$}
We would like to study the double scaling limit of ({\ref{n4}) as $N \to \infty$, and $M=\sqrt{2N}+kN^{-1/6}$ for some $k\in \R$.  With that in mind, we scale $M$ as $M=\sqrt{\frac{2N}{a}}$, and will study the limit as $N\to \infty$ and $a=1-LN^{-2/3}+O(N^{-4/3})$.  Then formulas (\ref{n4}) become
\begin{equation}\label{n4a}
\begin{aligned}
\mathbb{P} \bigg(\max_{0<t<1} b_N^{(BE)}(t)<M\bigg)&=\frac{a^{N^2+N/2}\pi^{2N^2+N/2}}{2^{N^2}N^{N^2+N/2}\prod_{k=0}^{N-1}(2k+1)!} \prod_{k=0}^{N-1} h_{2k+1}^{(0)}, \\
\mathbb{P} \bigg(\max_{0<t<1} b_N^{(R)}(t)<M\bigg)&=\frac{a^{N^2-N/2}\pi^{2N^2-3N/2}}{2^{N^2-N}N^{N^2-N/2}\prod_{k=0}^{N-1}(2k)!} \prod_{k=0}^{N-1} h_{2k}^{(1/2)}. \\
\end{aligned}
\end{equation}
As proved in Appendix \ref{defOP}, the products of normalizing constants in (\ref{n4a}) satisfy the deformation equations
\begin{equation}\label{n4b}
\begin{aligned}
\frac{d^2}{d a^2}\left(\log \prod_{k=0}^{N-1} h_{2k+1}^{(0)}\right)&=\left(\frac{\pi^2}{4N}\right)^2\frac{h_{2N+1}^{(0)}}{h_{2N-1}^{(0)}}\,, \\
\frac{d^2}{d a^2}\left(\log \prod_{k=0}^{N-1} h_{2k}^{(1/2)}\right)&=\left(\frac{\pi^2}{4N}\right)^2\frac{h_{2N}^{(1/2)}}{h_{2N-2}^{(1/2)}}\,. \\
\end{aligned}
\end{equation}
It follows that, if we denote
\begin{equation}\label{n4d}
F_N^{(BE)}(a)=\log\left[\mathbb{P} \left(\max_{0<t<1} b_N^{(BE)}(t)<\sqrt{\frac{2N}{a}}\right)\right], \quad F_N^{(R)}(a)=\log\left[\mathbb{P} \left(\max_{0<t<1} b_N^{(R)}(t)<\sqrt{\frac{2N}{a}}\right)\right]\,,
\end{equation}
then we have the deformation equations
\begin{equation}\label{n4e}
\frac{d^2}{d a^2} F_N^{(BE)}(a)=-\frac{2N^2+N}{2a^2}+\left(\frac{\pi^2}{4N}\right)^2\frac{h_{2N+1}^{(0)}}{h_{2N-1}^{(0)}}\,, \quad
\frac{d^2}{d a^2} F_N^{(R)}(a)=-\frac{2N^2-N}{2a^2}+\left(\frac{\pi^2}{4N}\right)^2\frac{h_{2N}^{(1/2)}}{h_{2N-2}^{(1/2)}}.
\end{equation}

In the scaling of $M$ described above, we can write the orthogonality condition  (\ref{n5}) as
\begin{equation}\label{rop2}
\sum_{x\in L_{n,\al}} P_k^{(\al)}(nx)P_m^{(\al)}(nx) e^{-nV(x)} =h_k^{(\al)} \de_{km}\,, \quad V(x)=\frac{\pi^2\xi a}{2}x^2\,, \quad \xi=\frac{n}{2N}.
\end{equation}
This is the same orthogonality condition as (\ref{dope2}) with $a \mapsto a\xi$, and thus using (\ref{dope4a}) we can write 
the formulas (\ref{n4e}) as
\begin{equation}\label{rop5}
\begin{aligned}
\frac{d^2}{d a^2} F_N^{(BE)}(a)=-\frac{2N^2+N}{2a^2}+\pi^4N^2\xi_+^{4N+3}\frac{h_{2N+1,2N+1}^{(0)}(a\xi_+)}{h_{2N,2N-1}^{(0)}(a)}\,, \qquad \xi_+=1+\frac{1}{2N}\,, \\
\frac{d^2}{d a^2} F_N^{(R)}(a)=-\frac{2N^2-N}{2a^2}+\pi^4N^2\xi_-^{-4N+3}\frac{h_{2N,2N}^{(1/2)}(a)}{h_{2N-1,2N-2}^{(1/2)}(a\xi_-)}\,, \qquad \xi_-=1-\frac{1}{2N}\,.
\end{aligned}
\end{equation}
Notice that as $a\to 0$, $M\to \infty$.  Since $M$ is typically close to $\sqrt{2N}$, it is reasonable to assume that both $F_N^{(BE)}(a)$ and $F_N^{(R)}(a)$  go to zero very quickly as $a \to 0$.  One might easily guess the following lemma.
\begin{lem}\label{smalla}
\begin{equation}\label{n21}
F_N^{(BE)}(0)=F_N^{(R)}(0)=\frac{d}{d a} F_N^{(BE)}(a) \bigg|_{a=0}=\frac{d}{d a} F_N^{(R)}(a) \bigg|_{a=0}=0.
\end{equation}
Furthermore, as $N\to \infty$, 
\begin{equation}\label{n20a}
\frac{d^2}{d a^2} F_N^{(BE)}(a) \bigg|_{a=0}=O\left(N^{-2}\right)\,, \qquad \frac{d^2}{d a^2} F_N^{(R)}(a) \bigg|_{a=0}=O\left(N^{-2}\right).
\end{equation}
\end{lem}
We leave the proof of this lemma the Appendix \ref{prooflem1}.  From (\ref{n21}) and (\ref{rop5}), we easily obtain the following integral representations for $F_N^{(BE)}(a)$ and $F_N^{(R)}(a)$.
\begin{prop}
The functions $F_N^{(BE)}(a)$ and $F_N^{(R)}(a)$, defined in (\ref{n4d}), have the integral representations
\begin{equation}\label{n19}
F_N^{(BE)}(a)=\int_0^a \int_0^u G^{(BE)}(r)\, dr\,du\,, \qquad F_N^{(R)}(a)=\int_0^a \int_0^u G^{(R)}(r) \,dr\,du\,, 
\end{equation}
where 
\begin{equation}
\begin{aligned}
G^{(BE)}(r)&=-\frac{2N^2+N}{2r^2}+\pi^4N^2\xi_+^{4N+3}\frac{h_{2N+1,2N+1}^{(0)}(r\xi_+)}{h_{2N,2N-1}^{(0)}(r)}\,, \\ G^{(R)}(r)&=-\frac{2N^2-N}{2r^2}+\pi^4N^2\xi_-^{-4N+3}\frac{h_{2N,2N}^{(1/2)}(r)}{h_{2N-1,2N-2}^{(1/2)}(r\xi_-)}\,.
\end{aligned}
\end{equation}

\end{prop}

\subsection{Evaluation of the integrals (\ref{n19})}
We would like to evaluate the integrals (\ref{n19}) in the limit as $N\to \infty$, and 
\begin{equation}\label{rop0}
a=1-LN^{-2/3}+O(N^{-4/3}).
\end{equation}
Let us write (\ref{n19}) as
\begin{equation}\label{rop7}
F_N^{(BE)}(a)=I_0^{(BE)}+I_1^{(BE)}+I_2^{(BE)}, \qquad F_N^{(R)}(a)=I_0^{(R)}+I_1^{(R)}+I_2^{(R)},
\end{equation}
where 
\begin{equation}\label{rop8}
\begin{aligned}
I_0^{(BE)}&= \int_0^{1-N^{-\ep}}\int_0^u G^{(BE)}(r) \, dr\,du \,, \qquad &I_1^{(BE)}&=&\int_{1-N^{-\ep}}^a \int_0^{1-N^{-\ep}} G^{(BE)}(r) \, dr\,du \,, \\ I_2^{(BE)}&=\int_{1-N^{-\ep}}^a \int_{1-N^{-\ep}}^u G^{(BE)}(r) \, dr\,du \,, \qquad &I_0^{(R)}&= &\int_0^{1-N^{-\ep}}\int_0^u G^{(R)}(r) \, dr\,du\,, \\
 I_1^{(R)}&=\int_{1-N^{-\ep}}^a \int_0^{1-N^{-\ep}} G^{(R)}(r) \, dr\,du \,, \qquad &I_2^{(R)}&=&\int_{1-N^{-\ep}}^a \int_{1-N^{-\ep}}^u G^{(R)}(r) \, dr\,du \,.
\end{aligned}
\end{equation}
for some $0<\ep<2/3$.

Consider first $I_0$.  We need a large $n$ formula for the normalizing constants $h_{n,k}^{(\al)}(r)$ when $r < 1-n^{-\ep}$.  In this case the asymptotics of $h_{n,k}^{(\al)}(r)$ match the asymptotics of the corresponding system of continuous orthogonal polynomials, the (monic and rescaled) Hermite polynomials.  We have the following lemma, whose proof is given in Appendix \ref{sub}.
\begin{lem}\label{subcrit}
Let $h_{n,k}^{(\al)}(a)$ be defined as in (\ref{dope2}) and (\ref{dope3}).  Let the parameter $a$ be such that $a<1-n^{-\ep}$ for some $0<\ep<2/3$.  Then as $n \to \infty$, 
 \begin{equation}\label{calc1}
 \begin{aligned}
h_{n,n}^{(\al)}(a)&=\frac{2}{\sqrt{a}}\left(\frac{1}{\pi^2 a e}\right)^n \left(1+\frac{1}{12n}+\frac{1}{288n^2} -\frac{139}{51840 n^3} +O(n^{-4})\right)  \,, \\
h_{n,n-1}^{(\al)}(a)^{-1}&=\frac{1}{2\pi^2\sqrt{a}}\left({\pi^2 a e}\right)^n\left(1-\frac{1}{12n}+\frac{1}{288n^2} +\frac{139}{51840 n^3} +O(n^{-4})\right) \,.
\end{aligned}
\end{equation}
\end{lem}
Combining these, we find that, as $N\to \infty$,
\begin{equation}\label{i02}
\begin{aligned}
\frac{h_{2N+1,2N+1}^{(0)}(r\xi_+)}{h_{2N,2N-1}^{(0)}(r)}&=\frac{1}{\pi^4 r^2 e^2}\left(1-\frac{1}{2N}+\frac{5}{24N^2}-\frac{1}{12N^3}+O(N^{-4})\right)\,, \\
\frac{h_{2N,2N}^{(1/2)}(r)}{h_{2N-1,2N}^{(1/2)}(r\xi_-)}&=\frac{1}{\pi^4 r^2 e^2}\left(1+\frac{1}{2N}+\frac{5}{24N^2}+\frac{1}{12N^3}+O(N^{-4})\right)\,,
\end{aligned}
\end{equation}
Inserting this asymptotic formula into the integrand of $I_0^{(BE)}$ and $I_0^{(R)}$ gives
\begin{equation}\label{i03}
\begin{aligned}
I_0^{(BE)}&= \int_0^{1-N^{-\ep}}\int_0^u \left(-\frac{2N^2+N}{2r^2}+\frac{N^2}{r^2}\left(1+\frac{1}{2N}+O(N^{-4})\right)\right)dr\,du=O(N^{-2})\,, \\
I_0^{(R)}&= \int_0^{1-N^{-\ep}}\int_0^u \left(-\frac{2N^2-N}{2r^2}+\frac{N^2}{r^2}\left(1-\frac{1}{2N}+O(N^{-4})\right)\right)dr\,du=O(N^{-2})\,,
\end{aligned}
\end{equation}
Similarly, we find that
\begin{equation}\label{i04}
I_1^{(BE)}=O(N^{-2-\ep})\,, \quad I_1^{(R)}=O(N^{-2-\ep})\,.
\end{equation}

We are left to evaluate $I_2^{(BE)}$ and $I_2^{(R)}$.  These integrals are in the regime in which Proposition \ref{asymptoticsh} is valid.
Let us write $r=1-xN^{-2/3}$, so that as $r$ varies from $1-N^{-\ep}$ to $a$, $x$ varies from $N^{2/3-\ep}$ to $L$.
Applying equation (\ref{ins5}) with $n=2N$, we obtain
\begin{equation}\label{int1}
\begin{aligned}
\frac{h_{2N+1,2N+1}^{(0)}(r\xi_+)}{h_{2N,2N-1}^{(0)}(r)}&=\frac{1}{\pi^2 e^2}\left(1+\frac{2^{5/3}(R'(s)+q'(s))}{N^{2/3}}+\frac{2x}{N^{2/3}}+O(N^{-1})\right)\,, \\
\frac{h_{2N,2N}^{(1/2)}(r)}{h_{2N-1,2N-2}^{(1/2)}(r\xi_-)}&=\frac{1}{\pi^2 e^2}\left(1+\frac{2^{5/3}(R'(s)+q'(s))}{N^{2/3}}+\frac{2x}{N^{2/3}}+O(N^{-1})\right).
\end{aligned}
\end{equation}
It follows that the integrands of $I_2^{(BE)}$ and $I_2^{(R)}$ agree up to the order $O(N)$.  According to (\ref{main3}) $s=2^{4/3}x+O(N^{-2/3})$.  
Inserting the formula (\ref{int1}) into either integrals $I_2^{(BE)}$ or $I_2^{(R)}$, we obtain
\begin{equation}\label{int2}
\begin{aligned}
I_2&=\int_{1-N^{-\ep}}^{1-LN^{-2/3}} \int_{1-N^{-\ep}}^u \left\{[-N^2-2xN^{4/3}+O(N)]+N^2\left(1+\frac{2^{5/3}T'(2^{4/3}x)}{N^{2/3}}\right.\right. \\
&\hspace{5cm}\left.\left.+\frac{2x}{N^{2/3}}+O(N^{-1})\right)\right\}dr\,du \\
&=\int_{1-N^{-\ep}}^{1-LN^{-2/3}} \int_{1-N^{-\ep}}^u \bigg(2^{5/3}N^{4/3}T'(2^{4/3}x)+O(N)\bigg) dr\,du,
\end{aligned}
\end{equation}
where $I_2$ is either of the integrals $I_2^{(BE)}$ and $I_2^{(R)}$.
If we write $u=1-yN^{-2/3}$, then we can write the integral (\ref{int2}) in terms of the variables $x$ and $y$:
\begin{equation}\label{int3}
\begin{aligned}
I_2=&\int_{L}^{N^{2/3-\ep}} \int_y^{N^{2/3-\ep}} \big(2^{5/3}N^{4/3}T'(2^{4/3}x)+O(N)\big)N^{-4/3}dx\, dy \\
=&\int_{L}^{N^{2/3-\ep}} \int_y^{N^{2/3-\ep}} \big(2^{5/3}T'(2^{4/3}x)+O(N^{-1/3})\big)dx\, dy.
\end{aligned}
\end{equation}
If we let $1/3<\ep<2/3$ then, after integrating, the error term goes to zero as $N\to \infty$, provided that it is uniform for large $x$.  This uniformity follows from the Riemann-Hilbert analysis, as discussed in section \ref{evalcrit}.  It follows that this integral has a limit as $N\to \infty$, which is 
\begin{equation}\label{int4}
\int_{L}^{\infty} \int_y^{\infty} 2^{5/3}T(2^{4/3}x)dx\, dy=-\int_L^\infty 2^{1/3} T(2^{4/3}y)dy=-\frac{1}{2}\int_{2^{4/3}L}^\infty T(x)dx.
\end{equation}
 Combining (\ref{i03}), (\ref{i04}), and (\ref{int4}), we thus find that, for $a=1-LN^{-2/3}$,
\begin{equation}\label{int5}
\begin{aligned}
F_N(a)=I_0+I_1+I_2&=-\frac{1}{2}\int_{2^{4/3}L}^\infty T_{2N+1}(x)dx+O(N^{-\de}) \\
&=-\frac{1}{2}\int_{2^{4/3}L}^\infty \bigg(R(x)+q(x)\bigg)dx+O(N^{-\de})\,, \qquad 0<\de<1/3.
\end{aligned}
\end{equation}
A simple change of variables and exponentiation gives (\ref{main21}), the result of Theorem \ref{maxheight}.  The proof of Theorem \ref{dopefreeenergy} is very similar, and we present it in the next section.

\medskip

\section{Proof of Theorem \ref{dopefreeenergy}}
A rescaling of (\ref{dope11}) using (\ref{dope4a}) gives the formula
\begin{equation}\label{fe3}
\begin{aligned}
\frac{\d}{\d a^2} F_n^{(DOPE)}(a)&=-\left(\frac{\pi^2}{2}\right)^2 \left(\frac{h_{n,n}(a)}{h_{n-1,n-2}\big(\xi_- a\big)}(\xi_-)^{-2n+3}+(\xi_+)^{2n+3}\frac{ h_{n+1,n+1}\big(\xi_+ a\big)}{h_{n,n-1}(a)}\right. \\
&\qquad\qquad \left.+ \frac{h_{n,n}^{(\al)}(a)}{h_{n,n-1}^{(\al)}(a)}\big(\xi_+ A_{n+1,n}^{(\al)} (a\xi_+)+A_{n,n-1}^{(\al)}(a)\big)^2 \right).
\end{aligned}
\end{equation}
The small $r$ behavior of the function $Z_n^{(DOPE)}(r)$ is described in the following lemma.
\begin{lem}\label{smallr}
As $r \to 0$, 
\begin{equation}\label{fe5}
F_n^{(DOPE)}(r)=\frac{\log(r)}{2}-\frac{1}{2}\log\left(\frac{2}{n\pi^2}\right)+F_n^{(GUE)}+O\left(\frac{r^2}{n^4}\right).
\end{equation}
where $F_n^{(GUE)}$ is defined in (\ref{main18}).
\end{lem}
The proof of this lemma is in Appendix \ref{prooflem1}. 
Consider now $F_n(a)$ for $a$ close to $1$.  According to (\ref{fe3}) and (\ref{fe5}), we have
\begin{equation}\label{fe6}
\begin{aligned}
F_n^{(DOPE)}(a)&=F_n^{(GUE)}+\frac{\log(a)}{2}-\frac{1}{2}\log\left(\frac{2}{n\pi^2}\right) +\int_0^a \int_0^u G^{(DOPE)}(r) \,dr\,du \\
&=F_n^{(GUE)}-\frac{1}{2}\log\left(\frac{2}{n\pi^2}\right)+\frac{\log(a)}{2}+I_0^{(DOPE)}+I_1^{(DOPE)}+I_2^{(DOPE)},
\end{aligned}
\end{equation}
where
\begin{equation}
\begin{aligned}
G^{(DOPE)}(r)&=\frac{1}{2r^2} -\frac{\pi^4}{4}\left(\xi_-^{-2n+3}\frac{h_{n,n}(r)}{h_{n-1,n-2}(\xi_- r)}+\xi_+^{2n+3}\frac{h_{n+1,n+1}(\xi_+r)}{h_{n,n-1}(r)}\right. \\
&\qquad\qquad\left. + \frac{h_{n,n}^{(\al)}(r)}{h_{n,n-1}^{(\al)}(r)}\big(\xi_+ A_{n+1,n}^{(\al)} (r\xi_+)+A_{n,n-1}^{(\al)}(r)\big)^2 \right)\,,
\end{aligned}
\end{equation}
and
\begin{equation}\label{fe7}
\begin{aligned}
I_0^{(DOPE)}&=\int_0^{1-n^{-\ep}} \int_0^u G^{(DOPE)}(r) \,dr\,du \,, \qquad
I_1^{(DOPE)}=\int_{1-n^{-\ep}}^a \int_0^{1-n^{-\ep}} G^{(DOPE)}(r)\,dr\,du \,,\\
I_2^{(DOPE)}&=\int_{1-n^{-\ep}}^a \int_{1-n^{-\ep}}^u G^{(DOPE)}(r)\,dr\,du,
\end{aligned}
\end{equation}
for some $0<\ep<2/3.$
In the regime $r< 1-n^{-\de}$ for some $0<\de<2/3$, the recurrence coefficients $A_{n,k}^{(\al)}$ are exponentially small in $n$, as shown in Appendix \ref{sub}.  Therefore, inserting the asymptotics (\ref{calc1}) into the integrals $I_0^{(DOPE)}$ and $I_1^{(DOPE)}$, we find that both $I_0^{(DOPE)}$ and $I_1^{(DOPE)}$ are $O(n^{-4})$ as $n\to \infty$.

We now evaluate $I_2^{(DOPE)}$.  We scale $a$ as $a=1-Ln^{-2/3}$and write $r=1-xn^{-2/3}$, so that as $r$ varies from $1-n^{-\ep}$, $x$ varies from $n^{2/3-\ep}$ to $L$.  Applying the asymptotics (\ref{main7a}), we find that in this regime, 
\begin{equation}
\xi_+ A_{n+1,n}^{(\al)} (r\xi_+)+A_{n,n-1}^{(\al)}(r)=\frac{(-1)^n 2^{5/3} \sin(2\pi\al)}{\pi n^{1/3}}\bigg(q(s)-q(s_+)\bigg)+O(n^{-2/3}),
\end{equation}
which, using (\ref{ins3}), is $O(n^{-2/3})$.  From (\ref{main6}), we have that $ \frac{h_{n,n}^{(\al)}(r)}{h_{n,n-1}^{(\al)}(r)}=O(1)$, and thus
\begin{equation}
\frac{h_{n,n}^{(\al)}(r)}{h_{n,n-1}^{(\al)}(r)}\big(\xi_+ A_{n+1,n}^{(\al)} (r\xi_+)+A_{n,n-1}^{(\al)}(r)\big)^2=O(n^{-4/3}).
\end{equation}
Applying the asymptotics (\ref{ins5}) the integral $I_2^{(DOPE)}$ can thus be written as
\begin{equation}
\int_{1-n^{-\ep}}^a \int_{1-n^{-\ep}}^u \bigg[-\frac{2^{1/3}}{n^{2/3}} \big(T_n'(s)+T_{n+1}'(s)\big)+O(n^{-1})\bigg] \, dr\, du.
\end{equation}
According to (\ref{main3}), $s=2^{2/3} x +O(n^{-2/3})$, and we therefore have
\begin{equation}
I_2^{(DOPE)}=\int_{1-n^{-\ep}}^{1-Ln^{-2/3}} \int_{1-n^{-\ep}}^u \left(\frac{2^{4/3}q^2(2^{2/3}x)}{n^{2/3}}+O(n^{-1})\right)dr\,du.
\end{equation}
Writing $u=1-yn^{-2/3}$ and taking $1/2<\ep<2/3$, this becomes
\begin{equation}
\begin{aligned}
I_2^{(DOPE)}&=\int_{L}^{n^{2/3-\ep}} \int_{y}^{n^{2/3-\ep}} \left(\frac{2^{4/3}q^2(2^{2/3}x)}{n^2}+O(n^{-7/3})\right)dx\,dy \\
&=\int_{L}^{n^{2/3-\ep}}  \left(\frac{2^{2/3}R(2^{2/3}y)}{n^2}+O(n^{-5/3-\ep})\right)\,dy \\
&=\int_{2^{2/3}L}^{\infty} \frac{R(x)}{n^2}\,dx+O(n^{-1-2\ep})\,. \\
\end{aligned}
\end{equation}
It follows that
\begin{equation}
\begin{aligned}
F_n^{(DOPE)}(a)-F_n^{(GUE)}-\frac{\log{a}}{2}+\frac{1}{2}\log\left(\frac{2}{n\pi^2}\right)&=\frac{1}{n^2} \int_{2^{2/3}L}^\infty R(x)\,dx +O(n^{-1-2\ep}) \\
&=-\frac{1}{n^2} \log\bigg(\fcal_2(2^{2/3}L)\bigg)+O(n^{-1-2\ep}),
\end{aligned}
\end{equation}
from which (\ref{main19}) follows immediately.

The rest of the paper is dedicated to the proof of Proposition \ref{asymptoticsh}, which is based on steepest descent analysis of a discrete Riemann-Hilbert problem.

\medskip 

\section{Riemann-Hilbert analysis}
\subsection{Equilibrium Measure and the $g$-function}\label{equilibrium}
The equilibrium measure associated with the weight $e^{-nV(x)}$ is the unique measure which minimizes the functional
\begin{equation}\label{eq1}
H(\nu)=\int \int \log\frac{1}{|x-y|}d\nu(x)d\nu(y)+\int V(x) d\nu(x),
\end{equation}
over the set of probability measures on $\R$.  In the case that $V(x)$ is given by (\ref{dope7}), it is well known that the solution to this equilibrium problem is supported on the interval $[-\frac{2}{\pi\sqrt{a}},\frac{2}{\pi \sqrt{a}}]$ and on this interval it has a density
\begin{equation}\label{eq2}
d\nu_0(x)=\frac{\pi a}{2} \sqrt{\frac{4}{\pi^2 a}-x^2} \ dx.
\end{equation}
Let us denote the density $\rho(x)$.
Clearly $\rho(x)$ has its maximum value at $x=0$ and $\rho(0)=\sqrt{a}$.  The critical value of the parameter $a$, for which $\rho(x)$ attains the upper constraint, is $a_c=1$.

We define the $g$-function associated with these orthogonal polynomials as the log transform of the equilibrium measure:
\begin{equation}\label{g1}
g(z)=\int_{-b}^b \log(z-x) \rho(x)dx, \quad b=\frac{2}{\pi \sqrt{a}},
\end{equation}
where we take the principal branch for the logarithm.  This function satisfies the following properties:

\begin{enumerate}
\item \( g(z)\) is analytic in \(\C\setminus
    (-\infty,b]\).
\item For large $z$,
\begin{equation}\label{g2}
 g(z)=\log z-\sum_{j=1}^\infty \frac{g_{2j}}{z^{2j}},\qquad 
g_{2j}=\int_{-b}^{b} \frac{x^{2j}}{2j}\,d\nu_0(x).
\end{equation}
\item The Euler-Lagrange variational conditions  for the equilibrium problem (\ref{eq1}) are
\begin{equation}\label{g5}
g_+(x)+g_-(x)\,
\left\{
\begin{aligned}
&=V(x)+l \quad \textrm{for}\quad x \in [-b,b] \\
&< V(x)+l \quad \textrm{for}\quad x \in \R \setminus [-b,b] ,
\end{aligned}
\right.
\end{equation}
where $g_+$ and $g_-$ refer to the limiting values from the upper and lower half-planes, respectively, and $l\in\R$ is the Lagrange multiplier.

\item The function
\begin{equation}\label{g6}
 G(x)\equiv g_+(x)-g_-(x)
\end{equation}
is pure imaginary for all real 
\( x\), and 
\begin{equation}\label{g7}
G(x)=2\pi i\int_x^{b}\rho(s)\,ds.
\end{equation} 

\item From (\ref{g5}) and (\ref{g7}) we obtain that
\begin{equation}\label{g9}
 2g_{\pm}(x)=
V(x)+l\pm 2\pi i \int_x^{b} \rho(s)ds\quad\textrm{for} \quad x \in [-b,b].
\end{equation}
\item 
Also from (\ref{g7}), we get that $G(x)$ is real analytic on the set $(-b,b)$.  We can therefore extend $G$ into a complex neighborhood of $(-b,b)$, and the Cauchy-Riemann equations imply that 
\begin{equation}\label{g10}
 \left. \frac{dG(x+iy)}{dy}\right|_{y=0}=2\pi \rho(x)\ge 0.
\end{equation}
\end{enumerate}
From (\ref{g5}) we have that
\begin{equation}\label{g11}
 G(x)=2g_+(x)-V(x)-l=-[2g_-(x)-V(x)-l],\quad x \in [-b,b].
\end{equation}


The value of the Lagrange multiplier is given by the equation
\begin{equation}\label{g12}
e^l=\frac{1}{\pi^2ae}\,.
\end{equation}

 \subsection{Interpolation Problem}\label{IP}
The orthogonal polynomials (\ref{dope2}) are encoded in the following interpolation problem (IP).  
For a given $n=0,1,\ldots$, find a $2\times 2$ matrix-valued function
$\bold P_n(z)=(\bold P_n(z)_{ij})_{1\le i,j\le 2}$ with the following properties:
\begin{enumerate}
\item
{\it Analyticity}: $\bold P_n(z)$ is an analytic function of $z$ for $z\in\C\setminus L_{n,\al}$.
\item
{\it Residues at poles}: At each node $x\in L_{n,\al}$, the elements $\bold P_n(z)_{11}$ and
$\bold P_n(z)_{21}$ of the matrix $\bold P_n(z)$ are analytic functions of $z$, and the elements $\bold P_n(z)_{12}$ and
$\bold P_n(z)_{22}$ have a simple pole with the residues,
\begin{equation} \label{IP1}
\underset{z=x}{\rm Res}\; \bold P_n(z)_{j2}=w_n(x)\bold P_n(x)_{j1},\quad j=1,2.
\end{equation}
\item
{\it Asymptotics at infinity}: There exists a function $r(x)>0$ on  $L_{n,\al}$ such that 
\begin{equation} \label{IP2a}
\lim_{x\to\infty} r(x)=0,
\end{equation} 
and such that as $z\to\infty$, $\bold P_n(z)$ admits the asymptotic expansion,
\begin{equation} \label{IP2}
\bold P_n(z)\sim \left( I+\frac {\bold P_1}{z}+\frac {\bold P_2}{z^2}+\ldots\right)
\begin{pmatrix}
z^n & 0 \\
0 & z^{-n}
\end{pmatrix},\qquad z\in \C\setminus \left[\bigcup_{x\in L_{n,\al}}^\infty D\big(x,r(x)\big)\right],
\end{equation}
where $D(x,r(x))$ denotes a disk of radius $r(x)>0$ centered at $x$ and $I$ is the identity matrix.
\end{enumerate}

It is not difficult to see (see \cite{BKMM}) that 
the IP has a unique solution, which is
\begin{equation} \label{IP3}
\bold P_n(z)=
\begin{pmatrix}
P_n(z) & n^{-1}C(w_n P_n)(z) \\
(h_{n-1})^{-1}P_{n-1}(z) & (nh_{n-1})^{-1}C(w_nP_{n-1})(z)
\end{pmatrix},
\end{equation}
where the Cauchy transformation $C$ is defined by the formula,
\begin{equation} \label{IP4}
C(f)(z)=\sum_{x\in L_{n,\al}}\frac{f(x)}{z-x}\,.
\end{equation}
Because of the orthogonality condition, as $z\to\infty$,
\begin{equation} \label{IP5}
\frac{1}{n}C(w_nP_n)(z)=\sum_{x\in L_{n,\al}}\frac{w_n(x)P_n(x)}{n(z-x)}
\sim \sum_{x\in L_{n,\al}} \frac{w_n(x)}{n}P_n(x)\sum_{j=0}^\infty \frac{x^j}{z^{j+1}}=\frac{h_n}{z^{n+1}}
+\sum_{j=n+2}^\infty \frac{a_j}{z^j},
\end{equation}
for some constants $a_j$, which justifies asymptotic expansion (\ref{IP2}), and we have that 
\begin{equation} \label{IP6}
h_{n,n}^{(\al)}=[\bold P_1]_{12},\qquad \left(h_{n,n-1}^{(\al)}\right)^{-1}=[\bold P_1]_{21}\,.
\end{equation}
Furthermore, the recurrence coefficient $A_{n,n-1}^{(\al)}$ is given by (see \cite{BKMM}, \cite{BL})
\begin{equation}\label{IP6a}
A_{n,n-1}^{(\al)}=\frac{[\bold P_2]_{21}}{[\bold P_1]_{21}}-[\bold P_1]_{11}
\end{equation}
and the Christoffel-Darboux kernel (\ref{dope13}) is given by
\begin{equation}\label{ope4a}
K_n(x,y)=\frac{e^{-nV(x)/2}e^{-nV(y)/2}}{n(x-y)} \begin{pmatrix} 0 & 1 \end{pmatrix} \bold P_n(x)^{-1}\bold P_n(y) \begin{pmatrix} 1 \\ 0 \end{pmatrix}.
\end{equation}

The asymptotic analysis of this IP follows the steepest descent method of Deift-Zhou, the plan of which is as follows.  We first convert the IP to a Riemann-Hilbert problem (RHP), where the condition on poles and residues is replaced by a jump condition on some contours in $\C$.  Then we perform a series of explicit transformations which convert the RHP to one with jumps which approach the identity matrix as $n \to \infty$.  This small norm problem can be solved by a series of perturbation theory.  We can then recover the orthogonal polynomials encoded in the IP by inverting the explicit transformations which led us to the small norm problem.

\subsection{Reduction of IP to RHP}\label{red}
We now reduce the interpolation problem to a Riemann-Hilbert problem.  Introduce the function
\begin{equation}\label{red1}
\Pi(z)=\frac{\sin(n\pi z+\al\pi)}{n\pi}.
\end{equation}
Notice that
\begin{equation}\label{red2}
\Pi(x_k)=0, \quad \Pi'(x_k)=\exp\left(in\pi x_k+i\pi\al\right)=(-1)^k, \quad \textrm{for} \quad x_k=\frac{k-\al}{n} \in L_{n,\al}.
\end{equation}

Introduce the upper triangular matrices,
\begin{equation} \label{red3}
\bold D^u_{\pm}(z)
=\begin{pmatrix}
1 & -\frac{w_n(z)}{n\Pi(z)}e^{\pm i\pi(n  z+\al)}   \\
0 & 1
\end{pmatrix},
\end{equation}
and the lower triangular matrices,
\begin{equation} \label{red4}
\bold D^l_{\pm}
=\begin{pmatrix}
\Pi(z)^{-1} & 0   \\
-\frac{n}{w_n(z)}e^{\pm i\pi(n z+\al)} & \Pi(z)
\end{pmatrix}
=\begin{pmatrix}
\Pi(z)^{-1} & 0   \\
0 & \Pi(z)
\end{pmatrix}
\begin{pmatrix}
1 & 0   \\
-\frac{n}{\Pi(z)w_n(z)}e^{\pm i\pi(n z+\al)} & 1
\end{pmatrix}.
\end{equation}
Define the matrix-valued functions,
\begin{equation} \label{red5}
\bold R^u_n=\bold P_n(z)\times
\left\{
\begin{aligned}
& \bold D^u_+(z) \quad \textrm{when}\quad \Im z\ge0 \\
& \bold D^u_-(z) \quad \textrm{when}\quad \Im z\le0,
\end{aligned}
\right.
\end{equation}
and 
\begin{equation} \label{red6}
\bold R^l_n=\bold P_n(z)\times
\left\{
\begin{aligned}
& \bold D^l_+(z),\quad \textrm{when}\quad \Im z\ge0 \\
& \bold D^l_-(z),\quad \textrm{when}\quad \Im z\le0.
\end{aligned}
\right.
\end{equation}
The functions $\bold R_n^u(z)$, $\bold R_n^l(z)$ are meromorphic on the closed upper and lower complex planes and they are two-valued on the real axis. Their possible poles
are located on the lattice $L_{n,\al}$. As shown in \cite{BL},
in fact they do not have any poles at all.

Consider the regions $\Om_{\pm}^\Delta$ and $\Om_{\pm}^\nabla$ shown in Figure \ref{contour_Sigma}, and the contour $\Sg$ which bounds these regions.  These regions lie entirely within an $\ep$-strip of the real line, and the regions $\Om_{\pm}^\Delta$ are small sectors above and below the origin, respectively, within this strip. 
\begin{figure}
\vskip 1cm
\scalebox{0.39}{\includegraphics{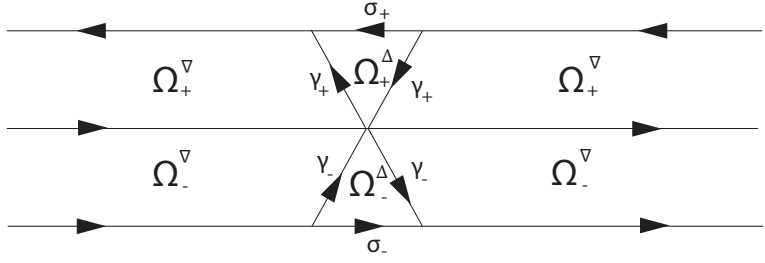}}
\caption{The contour $\Sigma$ dividing an $\ep$-neighborhood of the real line into the regions $\Om_\pm^\Delta$ and $\Om_\pm^\nabla$. }\label{contour_Sigma}
\label{fig2}
\end{figure}
 We make the transformation
\begin{equation} \label{red9}
\bold R_n(z)=
\left\{
\begin{aligned}
&\bold K \bold R_n^u(z)\bold K^{-1},\quad \textrm{for}\quad z\in\Om_{\pm}^\nabla,\\
&\bold K\bold R_n^l(z)\bold K^{-1},\quad \textrm{for}\quad z\in\Om_{\pm}^\De,\\
&\bold K \bold P_n(z)\bold K^{-1}, \quad \textrm{otherwise}.
\end{aligned}
\right.
\end{equation}
where $\bold K=\begin{pmatrix} 1 & 0 \\ 0 & -2i\pi \end{pmatrix}$.

Let us denote by $\ga_{\pm}$ the part of the contour $\Sg$ which is the boundary between the region $\Om_{\pm}^{\Delta}$ and the region $\Om_{\pm}^{\nabla}$.  Then the region $\Om_{\pm}^{\Delta}$ is bounded by the contour $\ga_{\pm}$ and a small segment on which $\Im z =\pm \ep$.  Let us denote these segments $\sg_{\pm}$.  


\subsection{First transformation of the RHP}\label{ft}

Define the matrix function $\bold T_n(z)$ as follows from the equation
\begin{equation}\label{ft1}
\bold R_n(z)=e^{\frac{nl}{2}\sigma_3}\bold T_n(z)e^{n(g(z)-\frac{l}{2})\sigma_3},
\end{equation}
where $l$ is the Lagrange multiplier, the function $g(z)$ is described 
in section \ref{equilibrium}, and $\sigma_3=\begin{pmatrix} 1 & 0 \\ 0 &-1 \end{pmatrix}$ is the third Pauli matrix.  
Then $\bold T_n(z)$ satisfies the following Riemann-Hilbert Problem:
\begin{enumerate}
  \item 
  $\bold T_n(z)$ is analytic in $\C \setminus \Sg$. 
  \item $\bold T_{n+}(z)=\bold T_{n-}(z)j_T(z)$ for $z\in\Sg$, where
  \begin{equation}\label{ft2}
   j_T(z)=
  \left\{
  \begin{aligned}
 & e^{n(g_-(z)-\frac{l}{2})\sigma_3}j_R(z)e^{-n(g_+(z)-\frac{l}{2})\sigma_3} \quad \textrm{for} \quad z \in \R \\
  &e^{n(g(z)-\frac{l}{2})\sigma_3}j_R(z)e^{-n(g(z)-\frac{l}{2})\sigma_3} \quad \textrm{for} \quad z \in \Sg \setminus \R.
  \end{aligned}\right.
  \end{equation}
  \item As $z \to \infty,$
   \begin{equation} \label{ft3}
\bold T_n(z)\sim  I+\frac {\bold T_1}{z}+\frac {\bold T_2}{z^2}+\ldots.
\end{equation}
\end{enumerate}

\subsection{Second transformation of the RHP}

Introduce the matrices
\begin{equation}\label{st1}
\begin{aligned}
j_-(z)&=\begin{pmatrix} 1 & 0 \\ e^{nG(z)} & 1 \end{pmatrix}, \quad &j_+(z)&=\begin{pmatrix} 1 & 0 \\ e^{-nG(z)} & 1 \end{pmatrix}, \\
\bold A_+(z) &=\begin{pmatrix} -\frac{1}{2n\pi i}e^{- i\pi(n z+\al)} & 0 \\ 0 & -2n\pi i e^{i\pi(n z+\al)} \end{pmatrix}, \quad &\bold A_-(z) &=\begin{pmatrix} \frac{1}{2n\pi i}e^{i\pi(n z+\al)} & 0 \\ 0 & 2n\pi i e^{-i\pi(n z+\al)} \end{pmatrix}\,.
\end{aligned}
\end{equation}
We now make the transformation
\begin{equation}\label{st2}
\begin{aligned}
&\bold S_n(z)=\left\{
\begin{aligned}
&\bold T_n(z)j_+(z)^{-1} \quad \textrm{for} \quad z\in \{(-b,b) \times (0,i\ep)\} \cap \Om_+^\nabla  \\
&\bold T_n(z)j_-(z)\quad \textrm{for} \quad z\in \{(-b,b) \times (0,-i\ep)\} \cap \Om_-^\nabla\\
&\bold T_n(z)\bold A_{\pm}(z)\quad \textrm{for} \quad z\in \Om_\pm^\Delta \\
&\bold T_n(z)  \quad \textrm{otherwise}.
\end{aligned}\right. \\ 
\end{aligned}
\end{equation}
This function satisfies a RHP similar to $\bold T$, but jumps now occur on a new contour $\Sigma_S$, 
which is obtained from $\Sg$ by adding the segments $(-b-i\ep, -b+i\ep)$ and $(b-i\ep, b+i\ep)$, see Figure \ref{sigma_S}.
\begin{figure}
\scalebox{0.3}{\includegraphics{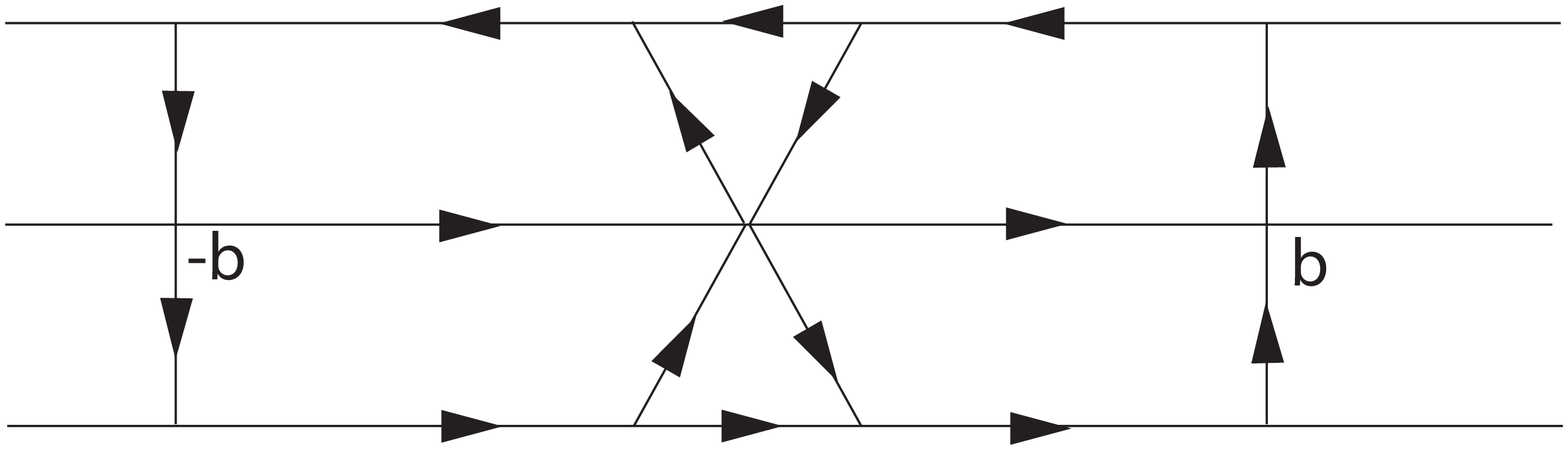}}
\caption{The contour $\Sigma_S$. }
\label{sigma_S}
\end{figure}
On this contour, the function $\bold S_n(z)$ satisfies the jump condition
\begin{equation}\label{st3}
\bold S_+(z)=\bold S_-(z) j_S(z),
\end{equation}
where
  \begin{equation}\label{st3a}
j_S(z)=
\left\{
\begin{aligned}
& \begin{pmatrix} 0 & 1 \\ -1 & 0 \end{pmatrix} \quad \textrm{when}\quad z\in (-b,b) \setminus \{0\} \\
& \begin{pmatrix} 1 & e^{n(g_+(z)+g_-(z)-l-V(z))} \\ 0 & 1 \end{pmatrix} \quad \textrm{when}\quad z\in \R\setminus(-b,b)  \\
&\begin{pmatrix}  
(1-e^{\pm 2i\pi(n z+\al)})^{-1} & \pm \frac{e^{\pm nG(z)}}{1-e^{\mp 2\pi i( n z+\al)}} \\ \mp e^{\mp nG(z)} & 1 \end{pmatrix} \quad \textrm{when}\quad z\in \{(-b,b) \pm i\ep\} \setminus \sg_{\pm}\\
&\begin{pmatrix}
(1-e^{\pm 2\pi i(n z+\al)})^{-1} & 0 \\ 
\mp  e^{\mp nG(z))} & 1-e^{\pm2\pi i(nz+\al)}
\end{pmatrix}
\quad \textrm{when }\quad z\in\sg_\pm\\
&\begin{pmatrix} 
1 &\mp e^{\pm nG(z)\pm 2\pi i(n z+\al)} \\ 
0
& 1
\end{pmatrix} 
\quad \textrm{when}\quad z\in  \ga_{\pm} \\
&\begin{pmatrix} 1 & \pm \frac{e^{n(2g(z)-l-V(z))}}{1-e^{\mp 2\pi i(nz+\al)}} \\ 0 & 1 \end{pmatrix} \quad z\in \bigg\{\{\R \setminus (-b,b)\} \pm i\ep \bigg\}\,.
\end{aligned}
\right.
\end{equation}

By formula (\ref{g5}) for the $G$-function and the upper constraint on the density $\rho$ we obtain that, for sufficiently small $\ep > 0$ and $x \in (-b,-\ep) \cup (\ep, b)$, 
\begin{equation}\label{st4}
0< \pm \Re G(x\pm i\ep) = 2\pi \ep \rho(x) +O(\ep^2) < 2\pi \ep +  O(\ep^2).
\end{equation}
Combined with property (\ref{g7}) of the $g$-function, this implies that all jumps on horizontal segments are exponentially close to the identity matrix, provided that they are bounded away from the segment $(-b,b)$.

\subsection{Model RHP}

The model RHP appears when we drop in the jump matrix $j_S(z)$ the terms that vanish as $n\to\infty$:
\begin{enumerate}
  \item 
  $\bold M(z)$ is analytic in $\C \setminus [-b,b]$.
  \item $\bold M_{+}(z)=\bold M_{-}(z)j_M(z)$ for $z\in[-b,b]$, where
  \begin{equation}\label{m1}
   j_M(z)=
  \begin{pmatrix} 0 & 1 \\ -1 & 0 \end{pmatrix}\,. 
 \end{equation}
 \item As $z \to \infty$,
\begin{equation} \label{m2}
\bold M(z)\sim I+\frac {\bold M_1}{z}+\frac {\bold M_2}{z^2}+\ldots.
\end{equation}
\end{enumerate}
This problem has the unique solution (see e.g., \cite{DKMVZ})
\begin{equation}\label{m3}
\bold M(z) = \begin{pmatrix} \frac{\ga(z)+\ga^{-1}(z)}{2} &  \frac{\ga(z)-\ga^{-1}(z)}{-2i} \\ \frac{\ga(z)-\ga^{-1}(z)}{2i} & \frac{\ga(z)+\ga^{-1}(z)}{2} \end{pmatrix}\,,
\end{equation}
where 
\begin{equation}\label{m4}
\ga(z)=\left(\frac{z+b}{z-b}\right)^{1/4}\,,
\end{equation}
with a cut on $[-b,b]$, taking the branch such that $\ga(z) \sim 1$ as $z\to \infty$.

\subsection{Parametrix at band-void edge points}

We now consider small disks $D(b,\ep)$ and $D(-b,\ep)$ around the endpoints of the support of the equilibrium measure.  Denote
\begin{equation}\label{pm0a}
D=D(b,\ep) \cup D(-b,\ep).
\end{equation}
We seek a local parametrix $\bold U(z)$ defined on $D$ such that
\begin{enumerate}
\item 
\begin{equation}\label{pm0}
\bold U(z) \ \textrm{is analytic on} \ D \setminus \Sigma_S.
\end{equation}
\item
\begin{equation}\label{pm1}
\bold U_{+}(z)=\bold U_{-}(z)j_S(z) \quad \textrm{for} \quad z\in D \cap \Sigma_S.
\end{equation}
\item
\begin{equation}\label{pm2}
\bold U(z)=\bold M(z) \big(I+O(n^{-1})\big) \quad \textrm{uniformly for} \  z\in \partial D.
\end{equation}
\end{enumerate}
The solution to the problem is given in \cite{DKMVZ} (see also \cite{BL}), and we do not repeat it here. 

\subsection{The Riemann-Hilbert problem associated with the Painlev\'{e} II equation}
In section \ref{pmo}, we will discuss the local analysis to our Riemann-Hilbert problem near the origin, which is the point at which the equilibrium measure attains the upper constraint.  The solution will be given in terms of a well known problem in integrable systems, which we discuss now.  For a more complete description of this problem, see the book \cite{FIKN}.

Let $\bold \Psi(\z)$ be the $ 2\times 2$ matrix-valued solution to the differential equation
\begin{equation}\label{pii1}
\frac{d}{d\z} \bold \Psi(\z)=\begin{pmatrix} -4i\z^2-i(s+2q^2) & 4\z q+2ir \\ 4\z q-2ir & 4i\z^2 +i(s+2q^2) \end{pmatrix}\bold \Psi(\z)\,.
\end{equation}
It is known that there exist solutions $\bold \Psi_j(\z)$ defined in each of the six sectors
\begin{equation}\label{pii2}
S_j=\left\{\z \in \C : \frac{2j-3}{6}\pi < \arg \z < \frac{2j-1}{6}\pi \right\},
\end{equation}
such that for $j \in \{1,2,3,4,5,6\}$, as $\z\to \infty$,
\begin{equation}\label{pii3}
 \bold \Psi_j(z)e^{i(\frac{4}{3}\z^3+s\z)\sg_3}=I+O(\z^{-1}),
 \end{equation}
and on the ray $\Ga_j=\{\z : \arg \z = \frac{2j-1}{6}\pi\}$ we have the jump condition
\begin{equation}\label{pii4}
\bold \Psi_{j+1}(\z)=\bold \Psi_j(\z) A_j,
\end{equation}
where
\begin{equation}\label{pii5}
A_j=\begin{pmatrix} 1 & 0 \\ a_j & 1 \end{pmatrix} \quad j \ \textrm{odd} \ , \quad A_j=\begin{pmatrix} 1 & a_j \\ 0 & 1 \end{pmatrix} \quad j \ \textrm{even}\,.
\end{equation}
The numbers $a_j$ are called the Stokes multipliers and satisfy the relations $a_{j+3}=a_j$, and $a_1a_2a_3+a_1+a_2+a_3=0$.

If the parameters $s$, $q$, and $r$ are chosen such that $q$, as a function of $s$, is the Hastings-McLeod solution to the Painlev\'{e} II equation and $r=q'(s)$, then the Stokes multipliers become 
\begin{equation}\label{pii7}
a_1=1, \quad a_2=0, \quad a_3=-1\,.
\end{equation}
It is known that the Hastings-McLeod solution to the Painlev\'{e} II equation is a meromorphic function whose (infinitely many) poles are all located away from the real line.  We will consider only real $s$, and thus we can always choose the parameters $r$, $q$, and $s$ in such a way.  Because $q(s)$ behaves like the Airy function for large $s>0$, in fact the error in (\ref{pii3}) is
\begin{equation}\label{pii7a}
O\left(\frac{e^{-\frac{2}{3}s^{3/2}}}{\z}\right)\,, \quad \textrm{as} \quad \z \to \infty, \quad s\to +\infty.
\end{equation}

Consider the oriented contour $\Gamma$ in the $\z$-plane depicted in Figure \ref{Gamma}, which is made up of the four rays $\ga_1, \ga_2, \ga_3$, and $\ga_4$, where
\begin{equation}\label{pii8}
\begin{aligned}
\ga_1&=\  \left\{\z \in \C : \arg \z  =\frac{\pi}{6} \right\}, \quad &\ga_2&=\  \left\{\z \in \C : \arg \z = \frac{5\pi}{6} \right\}, \\
\ga_3&=\  \left\{\z \in \C : \arg \z  =\frac{7\pi}{6} \right\}, \quad &\ga_4&=\  \left\{\z \in \C : \arg \z  =\frac{11\pi}{6} \right\}\,.
\end{aligned}
\end{equation}
\begin{figure}
\scalebox{0.3}{\includegraphics{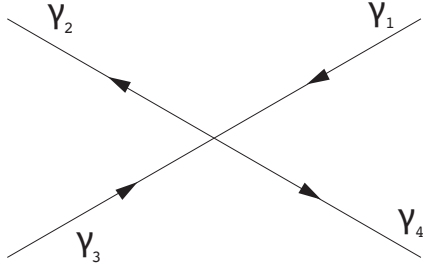}}
\caption{The contour $\Gamma$, which is composed of the rays $\ga_j$, $j=1,2,3,4.$}
\label{Gamma}
\end{figure}
The function $\bold \Psi(\z)$ which solves differential equation (\ref{pii1}) and corresponds to the Hastings-McLeod solution to the Painlev\'{e} II equation satisfies the following Riemann-Hilbert problem.
\begin{enumerate}
\item $\bold \Psi(\z)$ is analytic for $\z \in \C \setminus \Gamma.$
\item For $\z\in \Gamma$, $\bold \Psi(\z)$ satisfies the jump condition
\begin{equation}\label{pii9}
\bold \Psi_+(\z)=\bold \Psi_-(\z)j_{\bold \Psi}(\z),
\end{equation}
where
\begin{equation}\label{pii10}
j_{\bold \Psi}(\z)=\left\{
\begin{aligned} 
\begin{pmatrix} 1 & 0 \\ -1 & 1 \end{pmatrix} \quad \textrm{for} \ \z \in \ga_1 \cup \ga_2 \\
\begin{pmatrix} 1 & -1 \\ 0 & 1 \end{pmatrix} \quad \textrm{for} \ \z \in \ga_3 \cup \ga_4\,. \\
\end{aligned}\right.
\end{equation}
\item As $\z \to \infty,$
\begin{equation}\label{pii11}
\bold \Psi(\z) =(I+O(\z^{-1}))e^{-i(\frac{4}{3}\z^3+s\z)\sg_3}\,.
\end{equation}
\end{enumerate}

Introduce the parameter $\th \in \R$, and let $\bold \Psi_2(\z;s, \th)$ be defined via $\bold \Psi(\z)$ as 
\begin{equation}\label{pii13}
\bold \Psi_2(\z; s, \th)=\left\{
\begin{aligned}
&e^{i\th\sg_3} \sg_3 \bold \Psi(\z)e^{i(\frac{4}{3}\z^3+s\z)\sg_3} e^{-i\th\sg_3}\sg_1  \quad \Im \z >0 \\
&e^{i\th\sg_3}  \sg_3 \bold \Psi(\z) e^{i(\frac{4}{3}\z^3+s\z)\sg_3} e^{-i\th\sg_3} \sg_3 \quad \Im \z <0\,,
\end{aligned}\right.
\end{equation}
where $\sg_1=\begin{pmatrix} 0 & 1 \\ 1 & 0 \end{pmatrix}$ and $\sg_3=\begin{pmatrix} 1 & 0 \\ 0 & -1 \end{pmatrix}$ are Pauli matrices.
Then the function $\bold \Psi_2(\z; s, \th)$ satisfies the Riemann-Hilbert problem
\begin{enumerate}
\item $\bold \Psi_2(\z; s, \th)$ is analytic for $\z \in \C \setminus \big(\Gamma \cup \R\big).$
\item For $\z\in \Gamma$, $\bold \Psi_2(\z; s, \th)$ satisfies the jump condition
\begin{equation}\label{pii14}
\bold \Psi_{2+}(\z; s, \th)=\bold \Psi_{2-}(\z; s, \th)j_{2}(\z; s, \th),
\end{equation}
where
\begin{equation}\label{pii15}
j_2(\z; s, \th)=\left\{
\begin{aligned}
&\begin{pmatrix} 1 & -e^{2i(\frac{4}{3}\z^3+s\z-\th)\sg_3} \\ 0 & 1 \end{pmatrix} \quad \z \in \ga_1 \cup \ga_2 \\
&\begin{pmatrix} 1 & e^{-2i(\frac{4}{3}\z^3+s\z-\th)\sg_3} \\ 0 & 1 \end{pmatrix} \quad \z \in \ga_3 \cup \ga_4 \\
&\begin{pmatrix} 0 & 1 \\ -1 & 0 \end{pmatrix} \quad \z \in \R\,.
\end{aligned}\right.
\end{equation}
\item As $\z \to \infty,$
\begin{equation}\label{pii16}
\bold \Psi_2(\z; s, \th)=\left\{
\begin{aligned}
 & \big(I+O(\z^{-1})\big)\begin{pmatrix} 0 & 1 \\  -1 & 0 \end{pmatrix} \quad \Im \z >0 \\
 &\big(I+O(\z^{-1})\big)  \quad \Im \z <0 \,.
 \end{aligned}\right.
\end{equation}
\end{enumerate}

Finally we note, as in \cite{CK}, the formula for the entries of the matrix $\bold \Psi$,
\begin{equation}\label{pii21}
\bold \Psi(\z;s)\begin{pmatrix} 1 \\ 1 \end{pmatrix} = \begin{pmatrix} \Phi_1(\z;s) \\ \Phi_2(\z;s)\end{pmatrix} \quad \z \in S_1 \cup S_4,
\end{equation}
where $\Phi_{1,2}(\z;s)$  are defined in (\ref{main11}).

\subsection{Parametrix at the origin}\label{pmo}
We seek a local parametrix $\bold U(z)$ defined on $D(0,\ep)$ such that
\begin{enumerate}
\item 
$\bold U(z)$ is analytic for $\z \in D(0,\ep) \setminus \Sigma_S$.
\item For $\z\in \Gamma$, $\bold \Psi_2(\z)$ satisfies the jump condition
\begin{equation}\label{pmz1}
\bold U_{+}(z)=\bold U_{-}(z)j_S(z) \quad \textrm{for} \quad z\in D(0,\ep) \cap \Sigma_S.
\end{equation}
\item As $n \to \infty$,
\begin{equation}\label{pmz2}
\bold U(z)=\bold M(z) \big(I+O(n^{-1/3})\big) \quad \textrm{uniformly for} \  z\in \partial D(0,\ep).
\end{equation}
\end{enumerate}
Let us recall the jumps $j_S$ on the contours $\ga_+$, $\ga_-$, and $\R$ close to the origin:
\begin{equation}\label{pmz3}
j_{S}(z)=\left\{
\begin{aligned}
&\begin{pmatrix} 
1 &- e^{nG(z)+ 2\pi i(n z+\al)} \\ 
0
& 1
\end{pmatrix} \quad z\in \ga_+ \\
&\begin{pmatrix} 
1 & e^{- nG(z)- 2\pi i(n z+\al)} \\ 
0
& 1
\end{pmatrix} \quad z\in \ga_- \\
&\begin{pmatrix} 0 & 1 \\ -1 & 0 \end{pmatrix} \quad z\in \R \,.
\end{aligned}\right.
\end{equation}
These jumps are similar to the jumps given in (\ref{pii15}).  We thus seek $\z(z;t)$, $s(t)$, and $\th(t)$ which solve the equation
\begin{equation}\label{aa1}
2i\left(\frac{4}{3}\z(z;a)^3+s(a)\z(z;a)-\th(a)\right)=n(G(z)+2\pi iz)+2\pi i \al,
\end{equation}
or equivalently
\begin{equation}\label{aa2}
\left(\frac{1}{3}\z(z;a)^3+\frac{s(a)}{4}\z(z;t)-\frac{\th(a)}{4}\right)=\frac{n}{8i}(G(z)+2\pi iz)+\frac{\pi \al}{4}\,.
\end{equation}
As shown in \cite{CFU}, there is a unique solution to this equation which is regular at the origin, and it is defined in terms of the stationary points of the right side of (\ref{aa2}).  If we denote
\begin{equation}\label{aa3}
f(z;a)=\frac{n}{8i}(G(z)+2\pi iz)+\frac{\pi \al}{4}=\frac{\pi n}{8}\left[1+2z-2\int_0^z \rho(x)dx\right]+\frac{\pi \al}{4},
\end{equation}
then the zeroes of the function $f'(z;a)$ (here $'$ means differentiation with respect to $z$) are at the stationary points
\begin{equation}\label{aa4}
z_1(a)=\frac{2}{\pi a}\sqrt{a-1}, \quad z_2(a)=-\frac{2}{\pi a}\sqrt{a-1}\,.
\end{equation}
Notice that $z_1(1)=z_2(1)=0$.  We then have that (see \cite{CFU})
\begin{equation}\label{aa5}
\th(a)=-2\big(f(z_1;a)+f(z_2;a)\big), \quad s(a)^3=-36\big(f(z_2;a)-f(z_1;a)\big)^2.
\end{equation}
Notice that, as $\rho(x)$ is an even function of $x$, we have
\begin{equation}\label{aa7}
f(z;a)+f(-z;a)=\frac{n\pi}{4}+\frac{\pi \al}{2}, \quad f(z;a)-f(-z;a)=\frac{n\pi z}{2}-\frac{n\pi}{2}\int_0^z \rho(x) dx.
\end{equation}
It follows that 
\begin{equation}\label{aa8}
\th(a)=-\frac{n\pi}{2}-\pi \al, \quad s(a)^3=-9\pi^2n^2\left(z_1-\int_0^{z_1} \rho(x) dx\right)^2 .
\end{equation}

The possible zeroes of $\z(z;a)$ are the solutions to the equation
\begin{equation}\label{aa6}
f(z;a)=\frac{n\pi}{8}+\frac{\pi \al}{4},
\end{equation}
or equivalently,
\begin{equation}\label{aa10}
z=\int_0^z \rho(x)dx.
\end{equation}
The only solution to this equation, and thus the only possible zero of $\z(z)$, is at $z=0$.  Indeed, it is not difficult to see that $\z(0;a)=0$, and we thus have
\begin{equation}\label{aa11}
\z'(0;a)=\frac{4f'(0;a)}{s(a)}=\frac{\pi n(1-\sqrt{a})}{s(a)}, \quad \z''(0;a)=\frac{4f''(0;a)}{s(a)}=0.
\end{equation}

We now take 
\begin{equation}\label{pmz7}
\tilde{\bold U}(z;s,\th)=\left\{
\begin{aligned}
&\bold M(z)\begin{pmatrix} 0 & -1 \\  1 & 0 \end{pmatrix} \bold \Psi_2(\z(z);s,\th) \quad \Im z>0 \\
&\bold M(z) \bold \Psi_2(\z(z);s,\th) \quad \Im z<0\,. \\
\end{aligned}\right.
\end{equation}
This function is analytic for $z\in D(0,\ep)\setminus \Sigma_S$.  For $z \in \Sigma_S$, it satisfies the jump conditions
\begin{equation}\label{pmz8}
\tilde{\bold U}_+(z;s,\th)=\tilde{\bold U}_-(z;s,\th)j_U(z)
\end{equation}
where
\begin{equation}\label{pmz9}
j_U(z)=\left\{
\begin{aligned}
&\begin{pmatrix} 1 & -e^{2i(\frac{4}{3}\z(z)^3+s\z(z)-\th)\sg_3} \\ 0 & 1 \end{pmatrix} \quad z \in \ga_+  \\
&\begin{pmatrix} 1 & e^{-2i(\frac{4}{3}\z(z)^3+s\z(z)-\th)\sg_3} \\ 0 & 1 \end{pmatrix} \quad z \in \ga_-  \\
&\begin{pmatrix} 0 & 1 \\ -1 & 0 \end{pmatrix} \quad \z \in \R.
\end{aligned}\right.
\end{equation}
Let us check the large $n$ behavior of the function $\tilde{\bold U}(z).$  As $n\to \infty$, we have
\begin{equation}\label{pmz15}
\begin{aligned}
\tilde{\bold U}(z;s,\th)&=\left\{
\begin{aligned}
&\bold M(z)\begin{pmatrix} 0 & -1 \\  1 & 0 \end{pmatrix} (I+O(n^{-1/3}))\begin{pmatrix} 0 & 1 \\  -1 & 0 \end{pmatrix} \quad \Im z>0 \\
&\bold M(z)(I+O(n^{-1/3})) \quad \Im z<0\,. \\
\end{aligned}\right. \\
&=\bold M(z)(I+O(n^{-1/3}))\,.
\end{aligned}
\end{equation}
It follows that we may take our local solution to be $\bold U(z)=\tilde{\bold U}(z;s,\th)$, where $s=s(a)$ and $\th$ are given in (\ref{aa8}).

\subsection{The third and final transformation of the RHP}\label{tt}

We now consider the contour $\Sigma_X$, which consists of the circles $\partial D(-b, \ep)$, $\partial D(b, \ep)$, and $\partial D(0,\ep)$,  all oriented counterclockwise, together with the parts of 
$\Sigma_S \setminus [-b,b]$ 
which lie outside of the disks $D(-b, \ep)$, $D(b, \ep)$, and $D(0,\ep)$.
Let
\begin{equation}\label{tt1}
\bold X_n(z)=\left\{
\begin{aligned}
&\bold S_n(z) \bold M(z)^{-1} \quad \textrm{for} \ z \ \textrm{outside the disks }  D(-b, \ep), \ D(b, \ep), \ D(0,\ep) \\
&\bold S_n(z) \bold U(z)^{-1} \quad \textrm{for} \ z \ \textrm{inside the disks }  D(-b, \ep), \ D(b, \ep), \ D(0,\ep). \\
\end{aligned}\right.
\end{equation}
Then $\bold X_n(z)$ solves the following RHP:
\begin{enumerate}
\item   
$\bold X_n(z)$ is analytic on $\C \setminus \Sigma_X$.
\item
$\bold X_n(z)$ has the jump properties
\begin{equation}\label{tt2}
\bold X_{n+}(x)=\bold X_{n-}(z)j_X(z)
\end{equation}
where
\begin{equation}\label{tt3}
j_X(z)=\left\{
\begin{aligned}
&\bold M(z)\bold U(z)^{-1} \quad \textrm{for} \ z \ \textrm{on the circles} \\
&\bold M(z)j_S\bold M(z)^{-1} \quad \textrm{otherwise}.
\end{aligned}\right.
\end{equation}
\item
As $z\to\infty$, 
\begin{equation}\label{tt4}
\bold X_n(z)\sim I+\frac {\bold X_1}{z}+\frac {\bold X_2}{z^2}+\ldots
\end{equation}
\end{enumerate}
Additionally, we have that $j_X(z)$ is uniformly close to the identity in the following sense:
\begin{equation}\label{tt5}
j_X(z)=\left\{
\begin{aligned}
& I+O(n^{-1}) \quad \textrm{uniformly on the circles} \ \partial D(-b,\ep), \ \partial D(b,\ep) \\
& I+O(n^{-1/3}) \quad \textrm{uniformly on the circle} \ \partial D(0,\ep) \\
& I+O(e^{-C(z)n}) \quad \textrm{on the rest of} \ \Sigma_X, 
\end{aligned}\right.
\end{equation}
where $C(z)$ is a positive function bounded away from zero and with sufficient growth at infinity so that $e^{-C(z)}  \in L^1(\Sg_X)$.

If we set
\begin{equation}\label{tt6}
j_X^0(z)=j_X(z)-I,
\end{equation}
then (\ref{tt5}) becomes
\begin{equation}\label{tt7}
j_X^0(z)=\left\{
\begin{aligned}
&O(n^{-1}) \quad \textrm{uniformly on the circles} \  \partial D(-b,\ep), \ \partial D(b,\ep) \\
& O(n^{-1/3}) \quad \textrm{uniformly on the circle} \ \partial D(0,\ep) \\
&O(e^{-C(z)n}) \quad \textrm{on the rest of} \ \Sigma_X.
\end{aligned}\right.
\end{equation}

The solution to this small norm problem is given by a series of perturbation theory.  Namely, define the functions $v_k$ recursively as
\begin{equation}\label{tt15}
v_k(z)=-\frac{1}{2\pi i} \int_{\Sigma_X} \frac{v_{k-1}(u)j_X^0(u)}{z_--u}du\,, \quad  \quad v_0(z)=I,
\end{equation}
where $z_-$ means that the integration takes place on the minus-side of the contour.
The solution is then
\begin{equation}\label{tt17}
\bold X_n(z)=I+\sum_{k=1}^\infty\bold X_{n,k}(z)\,,
\end{equation}
where
\begin{equation}\label{tt18}
\bold X_{n,k}(z)=-\frac{1}{2\pi i}\int_{\Sigma_X}\frac{v_{k-1}(u)j_X^0(u)}{z-u}du.
\end{equation}
In particular, this implies that 
\begin{equation}\label{tt19}
\bold X_n \sim I + O\left(\frac{1}{n^{1/3}(|z|+1)}\right) \quad  \textrm{as} \quad n \to \infty 
\end{equation}
uniformly for $z\in \C \setminus \Sigma_X$.

\medskip

\section{Evaluation of $\bold X_{1}$ in the critical case}\label{evalcrit}
The function $\bold \Psi(\z)$ satisfies, as $\z \to \infty,$
\begin{equation}\label{eval3}
\begin{aligned}
e^{-i(\frac{4}{3}\z^3+s\z)\sg_3} \mathbf \Psi(\z)^{-1} =\left(I+\frac{1}{2i\z}A+\frac{1}{8\z^2}B +O(\z^{-3})\right)\,,
\end{aligned}
\end{equation}
where 
\begin{equation}\label{eval4}
A=-R(s)\sg_3+q(s)\sg_1\sg_3\,, \quad B=\big(q(s)^2-R(s)^2\big)I-2\big(q'(s)+q(s)R(s)\big)\sg_1\,,
\end{equation}
$q(s)$ is the Hastings-McLeod solution to Painlev\'{e} II, and
\begin{equation}\label{eval5}
R(s)=\int_s^\infty q(\xi)^2d\xi\,,
\end{equation}
see \cite{DZ1}.
For $ u \in \partial D(0,\ep)$, we have that 
\begin{equation}\label{eval6}
j_X(u)=\left\{
\begin{aligned}
&\bold M(u) \sg_1 e^{i\th \sg_3} e^{-i(\frac{4}{3}\z(u)^3+s(u)\z(u))\sg_3} \bold \Psi(\z(u);s)^{-1}  e^{-i\th \sg_3} \sg_1 \bold M(u)^{-1} \quad \Im u >0 \\
&\bold M(u) \sg_3 e^{i\th \sg_3} e^{-i(\frac{4}{3}\z(u)^3+s(u)\z(u))\sg_3} \bold \Psi(\z(u);s)^{-1} \sg_3 e^{-i\th \sg_3} \bold M(u)^{-1} \quad \Im u <0\,.
\end{aligned}\right.
\end{equation}
From (\ref{eval3}) and (\ref{eval6}), we see that
\begin{equation}
j_X^0(u)=\left\{
\begin{aligned}\label{eval7}
&\mathbf M(u) \sg_1 e^{i\th \sg_3} \left(\frac{1}{2i\z(u)}A+\frac{1}{8\z(u)^2}B +O(n^{-1})\right)e^{-i\th \sg_3}\sg_1\mathbf M(u)^{-1}  \\
& \hspace{9 cm} \Im u >0 \\
&\mathbf M(u) \sg_3 e^{i\th \sg_3} \left(\frac{1}{2i\z(u)}A+\frac{1}{8\z(u)^2}B +O(n^{-1})\right)\sg_3 e^{-i\th \sg_3} \mathbf M(u)^{-1} \\ & \hspace{9 cm} \Im u <0\,.
\end{aligned}\right.
\end{equation}
Notice that, according to (\ref{pii7a}), the error term in (\ref{eval7}) is uniform for $s \in [s_0, \infty)$ for any $s_0 \in \R$.  As can be seen in Appendix \ref{subcrit}, the jump matrices on the contours $D(\pm b, \ep)$ are uniformly close to the identity for $a$ close to 1.  Thus if we write $a=1-xn^{-2/3}$, using (\ref{main3}) and (\ref{tt17}), any error we compute in the large $n$ expansion of ${\bf X}_n$ will be uniform for $x \in [L, \infty)$ for any $L \in \R$.
Multiplying out the above expression, we get
\begin{equation}\label{eval8}
j_X^0(u)=\frac{R_1(u)}{2i\z(u)}+\frac{R_2(u)}{8\z(u)^2}+O(n^{-1}),
\end{equation}
where
\begin{equation}\label{eval9}
\begin{aligned}
R_1(u)&=\frac{R(s)}{2}\bigg(\tilde{\ga}(u)^2+\tilde{\ga}(u)^{-2}\bigg)\sg_3+\frac{R(s)}{2i}\bigg(\tilde{\ga}(u)^2-\tilde{\ga}(u)^{-2}\bigg)\sg_1+(-1)^n\cos(2\pi\al) q(s)\sg_3\sg_1 \\
& \qquad -\frac{(-1)^n}{2} \sin(2\pi\al)\bigg(\tilde{\ga}(u)^2-\tilde{\ga}(u)^{-2}\bigg)q(s) \sg_3 - \frac{(-1)^n}{2i}\sin(2\pi\al) \bigg(\tilde{\ga}(u)^2+\tilde{\ga}(u)^{-2}\bigg)q(s)\sg_1\,, \\
R_2(u)&=\big(q(s)^2-R(s)^2\big)I-i(-1)^n\cos(2\pi \al)\bigg(\tilde{\ga}(u)^2-\tilde{\ga}(u)^{-2}\bigg)\big(q'(s)+q(s)R(s)\big)\sg_3 \\
&\qquad -(-1)^n\cos(2\pi\al)\bigg(\tilde{\ga}(u)^2+\tilde{\ga}(u)^{-2}\bigg)\big(q'(s)+q(s)R(s)\big)\sg_1 \\
& \qquad -2i(-1)^n \sin(2\pi \al) \big(q'(s)+q(s)R(s)\big) \sg_3\sg_1.
\end{aligned}
\end{equation}
and $\tilde{\ga}(u)^2$ is the analytic continuation of $\ga(u)^2$ from the upper half plane.   The Taylor expansion of $\tilde{\ga}(u)^2$ about the origin is
\begin{equation}\label{eval10}
\tilde{\ga}(u)^2=-i\left(1+\frac{z}{b}+\frac{z^2}{2b}+\frac{z^3}{2b^3}+O(z^4)\right), \quad b=\frac{2}{\pi \sqrt{a}},
\end{equation}
which in turn gives that
\begin{equation}\label{eval11}
\tilde{\ga}(u)^2-\tilde{\ga}(u)^{-2}=-i\left(2+\frac{z^2}{b^2}+O(z^4)\right), \quad \tilde{\ga}(u)^2+\tilde{\ga}(u)^{-2}=-i\left(\frac{2z}{b}+\frac{z^3}{b^3}+O(z^5)\right).
\end{equation}

Expanding each term in (\ref{eval8}) about the origin gives
\begin{equation}\label{eval12}
\begin{aligned}
j_X^0(u)&=\frac{1}{2i\z'(0)}\left(\frac{R_1(0)}{u}+R_1'(0)+\left(\frac{R_1''(0)}{2}-\frac{\z'''(0)R_1(0)}{6\z'(0)}\right)u+\cdots \right) \\
& \qquad +\frac{1}{8\z'(0)^2}\left(\frac{R_2(0)}{u^2}+\frac{R_2'(0)}{u}+\frac{R_2''(0)}{2}-\frac{\z'''(0)R_2(0)}{3\z'(0)}+\cdots \right)+O(n^{-1})\,.
\end{aligned}
\end{equation}
Let us write this expansion as
\begin{equation}\label{eval13}
\begin{aligned}
j_X^0(u)&=\frac{1}{n^{1/3}}\left(\frac{A_{-1}}{u}+A_0+A_1 u+\cdots \right) +\frac{1}{n^{2/3}}\left(\frac{B_{-2}}{u^2}+\frac{B_{-1}}{u}+B_0+\cdots \right) +O(n^{-1}). \\
\end{aligned}
\end{equation}
In particular,
\begin{equation}\label{eval13a}
A_{-1}=\frac{n^{1/3}R_1(0)}{2i\z'(0)}\,, \quad A_{0}=\frac{n^{1/3}R_1'(0)}{2i\z'(0)}\,, \quad B_{-1}=\frac{n^{2/3}R_2'(0)}{8\z'(0)^2}\,, \quad B_{-2}=\frac{n^{2/3}R_2(0)}{8\z'(0)^2}\,.
\end{equation}

Let us evaluate $\mathbf X_{n,1}(z)$.  We have
\begin{equation}\label{eval14}
\mathbf X_{n,1}(z)=-\frac{1}{2\pi i}\int_{\partial D(0,\ep)} \frac{j_X^0(u)}{z-u} du+O(n^{-1}).
\end{equation}
This integral can be evaluated via a residue calculation.  Indeed, for $z \in \C \setminus D(0,\ep)$,
\begin{equation}\label{eval15}
\frac{1}{2\pi i} \int_{\d D(0,\ep)} \frac{j_X^0(u)du}{u-z}=-\frac{1}{z}\left[\frac{A_{-1}}{n^{1/3}}+\frac{1}{n^{2/3}}\left(B_{-1}+\frac{B_{-2}}{z}\right)\right]+O(n^{-1}).
\end{equation}
In particular, this gives that, for $z\in \Sg_X$,
\begin{equation}\label{eval16}
v_1(z)=-\frac{1}{z}\left[\frac{A_{-1}}{n^{1/3}}+\frac{1}{n^{2/3}}\left(B_{-1}+\frac{B_{-2}}{z}\right)\right]+O(n^{-1}),
\end{equation}
and thus
\begin{equation}\label{eval17}
\begin{aligned}
\mathbf X_{n,2}(z)&=\frac{1}{2\pi i } \int_{\d D(0,\ep)} \frac{v_1(u)j_X^0(u)}{u-z}du+O(n^{-4/3}) \\
&=-\frac{1}{2\pi i } \int_{\d D(0,\ep)} \frac{1}{(u-z)u}\left[\frac{A_{-1}}{n^{1/3}}\right]j_X^0(u)du+O(n^{-1}).
\end{aligned}
\end{equation}
For $z \in \C \setminus D(0,\ep)$ this is evaluated as
\begin{equation}\label{eval18}
\mathbf X_{n,2}(z)=\frac{1}{z}\left[\frac{1}{n^{2/3}}\left(A_{-1}A_0+\frac{A_{-1}^2}{z}\right)\right]+O(n^{-1}).
\end{equation}
This can also be taken as a formula for $v_2(z)$ for $z\in \Sg_X$.

Let $[\mathbf X_{n,k}]_j$ be the coefficient of the $z^{-j}$ term in the expansion of $\mathbf X_{n,k}(z)$ at $z=\infty$, so that
\begin{equation}\label{eval18a}
\mathbf X_{n,k}(z)=\frac{[\mathbf X_{n,k}]_1}{z}+\frac{[\mathbf X_{n,k}]_2}{z^2}+O(z^{-3}).
\end{equation}
From (\ref{eval16}) and (\ref{eval18}), we see that
\begin{equation}\label{eval19}
[\mathbf X_{n,1}]_1=-\frac{A_{-1}}{n^{1/3}}-\frac{B_{-1}}{n^{2/3}}+O(n^{-1}), \qquad
[\mathbf X_{n,2}]_1=\frac{A_{-1}A_0}{n^{2/3}}+O(n^{-1}).
\end{equation}
Adding these together gives that
\begin{equation}\label{eval20}
\begin{aligned}
\mathbf X_{1}&=[\mathbf X_{n,1}]_1+[\mathbf X_{n,2}]_1+O(n^{-1})\\
&=-\frac{A_{-1}}{n^{1/3}}+\frac{1}{n^{2/3}}(-B_{-1}+A_{-1}A_0)+O(n^{-1}) \\
&=\frac{1}{2\z'(0)}\left[iR_1(0)-\frac{R_2'(0)}{4\z'(0)}-\frac{R_1(0)R_1'(0)}{2\z'(0)}\right]+O(n^{-1}).
\end{aligned}
\end{equation}
Also from (\ref{eval16}) and (\ref{eval18}), we have
\begin{equation}\label{eval20a}
\begin{aligned}
\mathbf X_{2}&=[\mathbf X_{n,1}]_2+[\mathbf X_{n,2}]_2+O(n^{-1})\\
&=-\frac{B_{-2}}{n^{2/3}}+\frac{A_{-1}^2}{n^{2/3}}+O(n^{-1}) \\
&=-\frac{1}{8\z'(0)^2}\left[R_2(0)+2R_1(0)^2\right]+O(n^{-1}).
\end{aligned}
\end{equation}
Notice that
 \begin{equation}\label{eval23}
 \begin{aligned}
  R_1(0)&=-R(s)\sg_1+(-1)^n\cos(2\pi \al)q(s)\sg_3\sg_1+i(-1)^n \sin(2\pi \al) q(s)\sg_3, \\
R_1'(0)&=\frac{\pi \sqrt{a}R(s)}{2i}\sg_3+\frac{(-1)^n \pi \sqrt{a}}{2} \sin(2\pi\al) q(s) \sg_1, \\
R_2(0)&=\big(q(s)^2-R(s)^2\big)I-2(-1)^n\cos(2\pi\al)\big(q'(s)+q(s)R(s)\big)\sg_3 \\
& \qquad -2i(-1)^n\sin(2\pi\al)\big(q'(s)+q(s)R(s)\big)\sg_3\sg_1\,, \\
R_2'(0)&=(-1)^n i\pi \sqrt{a} \cos(2\pi \al)\big(q'(s)+q(s)R(s)\big)\sg_1.
  \end{aligned}
 \end{equation}
 
 \medskip
 
 \section{Proof of Proposition \ref{asymptoticsh} and Theorem \ref{kernel}}
 The quantities $h_{n,n}$ and $h_{n,n-1}$ are encoded in the matrix $\bold P_1$ described in (\ref{IP2}).  According to (\ref{red9}), (\ref{ft1}), and (\ref{tt1}), 
 \begin{equation}\label{as1}
\bold P_1=\bold K^{-1} \begin{pmatrix} [\bold X_1]_{11} + [\bold M_1]_{11} & \big([\bold X_1]_{12} + [\bold M_1]_{12}\big)e^{nl} \\ \big([\bold X_1]_{21} + [\bold M_1]_{21}\big)e^{-nl} &   [\bold X_1]_{22} + [\bold M_1]_{22}  \end{pmatrix} \bold K.
\end{equation}
It follows from (\ref{IP6}) that
\begin{equation}\label{as2}
\begin{aligned}
h_{n,n}&=-2\pi i\big([\bold X_1]_{12} + [\bold M_1]_{12}\big)e^{nl} \\
&=-2\pi i \left(\frac{1}{\pi^2 a e}\right)^n\left([\bold X_1]_{12} + \frac{ib}{2}\right)\\
&=\frac{2}{\sqrt{a}}\left(\frac{1}{\pi^2 a e}\right)^n\big(1-[\bold X_1]_{12}\pi i \sqrt{a} \big)  \,,
\end{aligned}
\end{equation}
and
\begin{equation}\label{as3}
\begin{aligned}
h_{n,n-1}^{-1}&=\frac{1}{-2\pi i}\big([\bold X_1]_{21} + [\bold M_1]_{21}\big)e^{-nl} \\
&=\frac{1}{-2\pi i} \left({\pi^2 a e}\right)^n\left([\bold X_1]_{21} - \frac{ib}{2}\right)\\
&=\frac{1}{2\sqrt{a}\pi^2}\left({\pi^2 a e}\right)^n\big(1+[\bold X_1]_{21}\pi i \sqrt{a} \big)  \,.
\end{aligned}
\end{equation}
According to (\ref{eval20}) and (\ref{eval23}),
\begin{equation}\label{as4}
\begin{aligned}
\left[\bold X_1\right]_{12}\pi i \sqrt{a} =& \frac{2^{2/3}}{n^{1/3}}\bigg(R-(-1)^n\cos(2\pi\al) q\bigg) \\
&\quad +\frac{2^{1/3}}{n^{2/3}}\bigg((-1)^n\cos(2\pi\al)(q'+2qR)-R^2+q^2\sin^2(2\pi\al)\bigg)+O(n^{-1}),
\end{aligned}
\end{equation}
and
\begin{equation}\label{as5}
\begin{aligned}
\left[\bold X_1\right]_{21}\pi i \sqrt{a} =& \frac{2^{2/3}}{n^{1/3}}\bigg(R+(-1)^n\cos(2\pi\al) q\bigg) \\
&\quad +\frac{2^{1/3}}{n^{2/3}}\bigg((-1)^n\cos(2\pi\al)(q'+2qR)+R^2-q^2\sin^2(2\pi\al)\bigg)+O(n^{-1}),
\end{aligned}
\end{equation}
where the functions $q$ and $R$ written with no argument refer to those functions evaluated at $s$.  This proves equation (\ref{main6}).

Also from (\ref{red9}), (\ref{ft1}), and (\ref{tt1}) we get
\begin{equation}
\frac{[\bold P_2]_{21}}{[\bold P_1]_{21}}-[\bold P_1]_{11}=\frac{[\bold X_2]_{21} + [\bold M_2]_{21}+[{\bf X}_1 {\bf M}_1]_{21}}{[\bold X_1]_{21}+[\bold M_1]_{21}}-[\bold X_1]_{11}-[\bold M_1]_{11}.
\end{equation}
Asymptotic evaluation of this expression from (\ref{eval20}) and (\ref{eval20a}), in light of (\ref{IP6a}), proves (\ref{main7a}).  Thus Proposition \ref{asymptoticsh} is proved.

We now turn to the proof of Theorem \ref{kernel}.  Consider $z\in D(0,\ep) \cap \Om^{\nabla}_{\pm}$.  In this region we have
\begin{equation}\label{ker1}
\bold P_n(z)=\left\{
\begin{aligned}
&\bold K^{-1} e^{\frac{nl}{2}\sg_3} \bold X_n(z)\bold M(z) \sg_1 e^{i\th\sg_3} \bold \Psi(\z(z)) e^{i(\z(z)^3+s\z(z;t)-\th)\sg_3} \sg_1 j_+(z) e^{n(g(z)-l/2)\sg_3} \\
& \qquad \times \bold K (\bold D_+^u(z))^{-1} \quad \Im z >0 \\
&\bold K^{-1} e^{\frac{nl}{2}\sg_3} \bold X_n(z)\bold M(z)e^{i\th\sg_3}\sg_3\bold \Psi(\z(z);s) e^{i(\z(z)^3+s\z(z;t)-\th)\sg_3} \sg_3 j_-(z)^{-1}  e^{n(g(z)-l/2)} \\
& \qquad \times \bold K (\bold D_-^u(z))^{-1} \quad \Im z <0.
\end{aligned}\right.
\end{equation}
Of course $\z(z)$ and $s$ depend on the parameter $a$, but we suppress the notation here.
Make the scaling 
\begin{equation}\label{ker2}
x=\frac{u}{cn^{1/3}},\quad y=\frac{v}{cn^{1/3}}, \quad c=\pi2^{-5/3}\,, \quad a=1-Ln^{-2/3}\,.
\end{equation}
Notice that in this scaling
\begin{equation}\label{ker3}
\z(x)=u+O(n^{-2/3}), \quad \z(y)=v+O(n^{-2/3}),\quad s=2^{2/3}L+O(n^{-2/3}),
\end{equation}
and that
\begin{equation}\label{ker4}
s(1)=s_\infty=\lim_{n\to\infty}s(t)=2^{2/3}L\,.
\end{equation}
It then follows that
\begin{equation}\label{ker5}
\bold P_n(x)^{-1} \bold P_n(y)=\left\{
\begin{aligned}
&\bold D_+^u(x) \bold K^{-1} e^{-n(g(x)-l/2)\sg_3} j_+(x)^{-1} \sg_1 e^{-i(\frac{4}{3}\z(x)^3+s\z(x)-\th)\sg_3} \bold \Psi(\z(x);s)^{-1} \\
&\quad \times \left(I+O\left(\frac{u-v}{n^{1/3}}\right)\right)\bold \Psi(\z(y);s)e^{i(\frac{4}{3}\z(y)^3+s\z(y)-\th)\sg_3}\sg_1 j_+(y) \\
&\quad \times e^{n(g(y)-l/2)\sg_3}\bold K \bold D_+^u(y)^{-1}\quad \Im x,y>0 \\
&\bold D_-^u(x) \bold K^{-1} e^{-n(g(x)-l/2)\sg_3} j_-(x) \sg_3 e^{-i(\frac{4}{3}\z(x)^3+s\z(x)-\th)\sg_3} \bold \Psi(\z(x);s)^{-1} \\
&\quad \times \left(I+O\left(\frac{u-v}{n^{1/3}}\right)\right)\bold \Psi(\z(y);s)e^{i(\frac{4}{3}\z(y)^3+s\z(y)-\th)\sg_3}\sg_3 j_-(y)^{-1} \\
&\quad \times e^{n(g(y)-l/2)\sg_3}\bold K \bold D_-^u(y)^{-1}\quad \Im x,y<0 \,.\\
\end{aligned}\right.
\end{equation}
Taking limits from either the upper or lower half planes, we find that for $x$ and $y$ real and in $D(0,\ep)$,
\begin{equation}\label{ker6}
\begin{aligned}
K_n(x,y)&=\frac{e^{-nV(x)/2}e^{-nV(y)/2}}{n(x-y)} \begin{pmatrix} 0 & 1 \end{pmatrix} \bold P_n(x)^{-1}\bold P_n(y) \begin{pmatrix} 1 \\ 0 \end{pmatrix} \\
&=\frac{cn^{1/3}}{2in\pi(u-v)} \begin{pmatrix} -e^{\frac{nG(x)}{2}} & e^{-\frac{nG(x)}{2}} \end{pmatrix} e^{-i(\frac{4}{3}\z(x)^3+s\z(x)-\th)\sg_3} \bold \Psi(\z(x);s)^{-1} \\
&\quad \times \left(I+O\left(\frac{u-v}{n^{1/3}}\right)\right)\bold \Psi(\z(y);s)e^{i(\frac{4}{3}\z(y)^3+s\z(y)-\th)\sg_3} \begin{pmatrix} e^{-\frac{nG(y)}{2}} \\ e^{\frac{nG(y)}{2}}\end{pmatrix} \\
&=\frac{cn^{1/3}}{2in\pi(u-v)} \begin{pmatrix} -e^{-i\pi (n x+\al)} & e^{i\pi(n x+\al)} \end{pmatrix} \bold \Psi(\z(x);s)^{-1} \left(I+O\left(\frac{u-v}{n^{1/3}}\right)\right) \\
& \quad \times \bold\Psi(\z(y);s) \begin{pmatrix} e^{i\pi(n y+\al)} \\ e^{-i\pi(n y+\al)}\end{pmatrix}.
\end{aligned}
\end{equation}
In order for this kernel to make sense, we must choose $x$ and $y$ to lie on the lattice $L_{n,\al}$.  If we take $u\neq v$ and
\begin{equation}\label{ker7}
x=\frac{k_n-\al}{n}, \quad y=\frac{m_n-\al}{n}; \quad k_n, m_n \in \Z,
\end{equation}
then the above expression becomes
\begin{equation}\label{ker8}
K_n(x,y)=\frac{(-1)^{k_n+m_n}cn^{1/3}}{2in\pi(u-v)} \begin{pmatrix} -1 & 1 \end{pmatrix} \bold \Psi\big(\z(x);s)\big)^{-1} \left(I+O\left(\frac{u-v}{n^{1/3}}\right)\right)\bold \Psi\big(\z(y);s\big) \begin{pmatrix} 1\\ 1\end{pmatrix}.
\end{equation}
According to (\ref{pii21}), 
\begin{equation}\label{ker9}
\bold \Psi(\z;s)\begin{pmatrix} 1 \\ 1 \end{pmatrix} = \begin{pmatrix}  \Phi_1(\z;s) \\  \Phi_2(\z;s) \end{pmatrix} , \quad \begin{pmatrix} -1 & 1 \end{pmatrix} \bold \Psi^{-1}(\z;s) = \begin{pmatrix} -\Phi_2(\z;s) & \Phi_1(\z;s) \end{pmatrix},
\end{equation}
thus (\ref{ker8}) implies (\ref{main15}).  Notice that the presence of the factors $e^{\pm i\pi(n z+\al)}$ in (\ref{ker6}) plays little role in the limiting values of the off diagonal terms.  To compute the diagonal terms, we must take a limit of (\ref{ker6}) as $u\to v$, and these factors do indeed play a role.  In taking this limit, we must take into account the fact that $\det \bold \Psi =1$, and we obtain (\ref{main17}), which proves Theorem \ref{kernel}.

\appendix

\section{Proof of Lemmas \ref{smalla} and \ref{smallr}}\label{prooflem1}
The proofs of Lemmas \ref{smalla} and \ref{smallr} are based on the following estimate of the rate of convergence of a Riemann sum, which is slightly sharper than the a priori rate of $O(\ep)$.
\begin{lem}\label{aplem}  Let $f(x_1,\dots, x_n)$ be an analytic function of $n$ variables.  Let the functions $A_k(x_1,\dots, x_n)$ be defined recursively via
\begin{equation}\label{ap9}
A_{k}(x_1,\dots,x_n)=\sum_{j=1}^n \frac{\d}{\d x_j} A_{k-1}(x_1,\dots,x_n)\,, \quad A_0 = f.
\end{equation}
Suppose that 
\begin{equation}\label{ap11}
\int \cdots \int_{\R^n} A_1(x_1, \dots, x_n) dx_1 \cdots dx_n=\int \cdots \int_{\R^n} A_2(x_1, \dots, x_n) dx_1 \cdots dx_n=0.
\end{equation}
Then as $\ep \to 0$,
\begin{equation}\label{ap13a}
\int \cdots \int_{\R^n} f(x_1, \dots, x_n) dx_1 \cdots dx_n-\ep^n\sum_{x_1, \dots, x_n \in \ep\Z}f(x_1, \dots, x_n)=O(\ep^4).
\end{equation}
\end{lem}

\begin{proof}
A multi-integral my be estimated by writing
\begin{equation}\label{ap6}
\int \cdots \int_{\R^n} f(x_1, \dots, x_n) dx_1 \cdots dx_n=\sum_{x_1, \dots, x_n \in \ep\Z} \int_{x_1}^{x_1+\ep} \cdots \int_{x_n}^{x_n+\ep}f(t_1, \dots, t_n) dt_1 \cdots dt_n.
\end{equation}
Expanding each integrand on the RHS as a Taylor series and integrating term by term, this becomes
\begin{equation}\label{ap7}
\sum_{x_1, \dots, x_n \in \ep\Z} \bigg[\ep^n f+\frac{A_1}{2}\ep^{n+1}+\frac{A_2}{6}\ep^{n+2}+\frac{A_3}{24}\ep^{n+3}+O(\ep^{n+4})\bigg],
\end{equation}
where
\begin{equation}\label{ap8}
f\equiv f(x_1,\dots,x_n)\,, \quad A_k\equiv A_j(x_1, \dots, x_n)\,.
\end{equation}
Thus the error in the Riemann sum is given by
\begin{equation}\label{ap10}
\begin{aligned}
\int \cdots& \int_{\R^n} f dx_1 \cdots dx_n-\sum_{x_1, \dots, x_n \in \ep\Z}\ep^n f=\ep^n \sum_{x_1, \dots, x_n \in \ep\Z} \bigg[\frac{A_1}{2}\ep+\frac{A_2}{6}\ep^{2}+\frac{A_3}{24}\ep^{3}+O(\ep^{4})\bigg] \\
&=\left(\ep^n \sum_{x_1, \dots, x_n \in \ep\Z} A_1\right)\frac{\ep}{2}+\left(\ep^n \sum_{x_1, \dots, x_n \in \ep\Z} A_2\right)\frac{\ep^2}{6}+\left(\ep^n \sum_{x_1, \dots, x_n \in \ep\Z} A_3\right)\frac{\ep^3}{24}+O(\ep^4).
\end{aligned}
\end{equation}
If (\ref{ap11}) holds, then by the same argument,
\begin{equation}\label{ap12}
\begin{aligned}
\ep^n\sum_{x_1, \dots, x_n \in \ep\Z} A_1(x_1, \dots, x_n)&=-\left(\ep^n \sum_{x_1, \dots, x_n \in \ep\Z} A_2\right)\frac{\ep}{2}-\left(\ep^n \sum_{x_1, \dots, x_n \in \ep\Z} A_3\right)\frac{\ep^2}{6}+O(\ep^3), \\
\ep^n\sum_{x_1, \dots, x_n \in \ep\Z} A_2(x_1, \dots, x_n)&=-\left(\ep^n \sum_{x_1, \dots x_n \in \ep\Z} A_3\right)\frac{\ep}{2}+O(\ep^2),
\end{aligned}
\end{equation}
and (\ref{ap10}) can be written as
\begin{equation}\label{ap13}
\begin{aligned}
\int \cdots& \int_{\R^n} f(x_1, \dots, x_n) dx_1 \cdots dx_n-\ep^n\sum_{x_1, \dots x_n \in \ep\Z}f(x_1, \dots, x_n) \\
&=-\left(\ep^n \sum_{x_1, \dots, x_n \in \ep\Z} A_2\right)\frac{\ep^2}{12}-\left(\ep^n \sum_{x_1, \dots, x_n \in \ep\Z} A_3\right)\frac{\ep^3}{24}+O(\ep^4)
&=O(\ep^4).
\end{aligned}
\end{equation}
\end{proof}

Let us first apply this result to the proof of Lemma \ref{smallr}.  By rescaling we can write (\ref{dope6}) as
\begin{equation}\label{ap20}
Z_n^{(DOPE)}(r)=\left(\frac{2}{\pi^2 n r}\right)^{n^2/2} \left(\ep^n \sum_{x_1, \dots, x_n \in \ep \{\Z-\al\}} \prod_{1\le j < k \le n} (x_k -x_j)^2 \exp\left\{-\sum_{j=1}^n x_j^2\right\}\right)
\end{equation}
where 
\begin{equation}\label{ap21}
\ep=\pi \sqrt{\frac{r}{2n}}.
\end{equation}
We have here an explicit prefactor times a Riemann sum for the function
\begin{equation}
f(x_1, \dots, x_n)=\prod_{1\le j < k \le n} (x_k -x_j)^2 \exp\left\{-\sum_{j=1}^n x_j^2\right\},
\end{equation}
which is exactly the integrand in (\ref{main18}).
One may easily check then that
\begin{equation}\label{ap22}
\begin{aligned}
A_1(x_1, \dots, x_n)&=-2(x_1+\cdots+x_n)f(x_1, \dots, x_n)\,, \\
 A_2(x_1, \dots, x_n)&=2f(x_1, \dots, x_n)\bigg(2(x_1+\cdots+x_n)^2-n\bigg) \\
 &=-2\bigg(A_1(x_1, \dots, x_n)(x_1+\cdots+x_n)+nf(x_1,\dots,x_n)\bigg).
 \end{aligned}
\end{equation}
It is easy to see that $A_1$ has the symmetry $A_1(x_1, \dots, x_n)=-A_1(-x_1, \dots, -x_n)$, from which it follows that
\begin{equation}\label{ap23}
\int \cdots \int_{\R^n} A_1(x_1, \dots, x_n) dx_1 \cdots dx_n=0.
\end{equation}
It is a simple exercise to integrate by parts to see that
\begin{equation}\label{ap14}
\int \cdots \int_{\R^n} A_2(x_1, \dots, x_n) dx_1 \cdots dx_n=0.
\end{equation}
It then follows that as $r\to 0$,
\begin{equation}\label{app25}
\begin{aligned}
Z_n^{(DOPE)}(r)&=\left(\frac{2}{\pi^2 n r}\right)^{n^2/2} Z_n^{(GUE)} \bigg(1+O(\ep^4)\bigg) \\
&=\left(\frac{2}{\pi^2 n r}\right)^{n^2/2} Z_n^{(GUE)} \left(1+O\left(\frac{r^2}{n^2}\right)\right).
\end{aligned}
\end{equation}
Taking logarithms gives (\ref{fe5}).

We now prove Lemma \ref{smalla} in the absorbing case.  The proof in the reflecting case is nearly identical.  Using symmetry about the origin and a rescaling of (\ref{n1}), we get
\begin{equation}\label{ap1}
\begin{aligned}
\mathbb{P} \bigg(\max_{0<t<1} &b_N^{(BE)}(t)<M\bigg) = \\
&\frac{2^{N(N+1)}}{N! \pi^{N/2} \prod_{k=0}^{N-1} (2k+1)!} \left(\ep^N \sum_{{\bf x} \in (\ep \N)^N} \big(\De({\bf x^2})\big)^2 \left(\prod_{j=1}^N x_j^2\right)\exp\left\{-\sum_{j=1}^N x_j^2\right\}\right) ,
\end{aligned}
\end{equation}
where 
\begin{equation}\label{ap2}
\ep=\frac{\pi}{M\sqrt{2}}.
\end{equation}
We again have an explicit prefactor times a Riemann sum for the integral
\begin{equation}\label{ap3}
\int_0^\infty \cdots \int_0^\infty \big(\De({\bf x^2})\big)^2 \left(\prod_{j=1}^N x_j^2\right)\exp\left\{-\sum_{j=1}^N x_j^2\right\}dx_1\cdots dx_N.
\end{equation}
This integral is the partition function for the Laguerre unitary ensemble.  Its value is known (see e.g., \cite{For}), and it exactly cancels the prefactor, so that
\begin{equation}\label{ap5}
\lim_{M \to \infty} \mathbb{P} \bigg(\max_{0<t<1} b_N(t)<M\bigg)=1.
\end{equation}
Lemma \ref{aplem} also holds for multi-integrals over $\R_+^n$, that is if we replace $\R^n$ with $\R_+$ and $\Z$ with $\N$,  and we can thus use Lemma \ref{aplem} with 
\begin{equation}\label{ap14b}
f(x_1, \dots, x_N)=\big(\De({\bf x^2})\big)^2 \left(\prod_{j=1}^N x_j^2\right)\exp\left\{-\sum_{j=1}^N x_j^2\right\}.
\end{equation}
It is not difficult to see that in this case the condition (\ref{ap11}) is satisfied.  Indeed, notice that for any $j=1,2,\dots,N$,
\begin{equation}\label{ap14a}
\int_0^\infty \left(\frac{\d}{\d x_j} f(x_1, \dots, x_N)\right) dx_j = -f((x_1,\dots,x_{j-1},0,x_{j+1},\dots,x_N)=0.
\end{equation}
Furthermore, notice that
\begin{equation}\label{ap16}
\frac{\d}{\d x_j} f(x_1, \dots, x_N)=\exp\left\{-\sum_{k=1}^N x_k^2\right\} \left( \prod_{k=1}^N x_k\right) P(x_1, \dots, x_N)
\end{equation}
for some polynomial $P$.  It follows that, for any $j=1,2,\dots,N$, 
\begin{equation}\label{ap16a}
A_1(x_1,\dots,x_{j-1},0,x_j,\dots,x_N)=0,
\end{equation}
and thus
\begin{equation}\label{ap17}
\int_0^\infty \left(\frac{\d}{\d x_j} A_1(x_1, \dots, x_N)\right) dx_j = -A_1(x_1,\dots,x_{j-1},0,x_j,\dots,x_N)=0.
\end{equation}
(\ref{ap14a}) and (\ref{ap17}) imply (\ref{ap11}), and thus Lemma \ref{aplem} applies.
It follows that
\begin{equation}\label{ap18}
\mathbb{P} \bigg(\max_{0<t<1} b_N^{(BE)}(t)<M\bigg) =1+O(\ep^4)=1+O(M^{-4}).
\end{equation}
In the scaling $M=\sqrt{\frac{2N}{a}}$ this becomes, as $a\to 0$,
\begin{equation}\label{ap19}
\mathbb{P} \bigg(\max_{0<t<1} b_N^{(BE)}(t)<M\bigg) =1+O(\ep^4)=1+O\left(\frac{a^2}{N^2}\right).
\end{equation}
Taking the logarithm proves Lemma \ref{smalla}.

\medskip

\section{Proof of Lemma \ref{subcrit}}\label{sub}
 If the parameter $a$ is such that
 \begin{equation}
 a<1-n^{-\de}, \quad 0<\de<\frac{2}{3},
 \end{equation}
 then by (\ref{main3}) and (\ref{pii7a}) the jump matrix for $\bold X_n(z)$ about the origin is exponentially small in $n$, and therefore the asymptotic expansion for $\bold X_n(z)$ comes from the jumps on the circles $\d D(\pm b,\ep)$.  We need to calculate this expansion up to an error of the order $n^{-3}$.  Instead of doing this directly, which is rather tedious, let us proceed by comparing our discrete orthogonal polynomials with their continuous brethren, the monic scaled Hermite polynomials $\{P_j^{(c)}(x)\}_{j=0}^\infty$, for which we have exact formulas.  These polynomials satisfy the orthogonality condition
\begin{equation}
\int_{-\infty}^{\infty} P_j^{(c)}(x)P_k^{(c)}(x)e^{-n\frac{a\pi^2x^2}{2}}\,dx=h_k^{(c)} \de_{jk}.
\end{equation}
The superscript $(c)$ stands for continuous. 
 The continuous orthogonal polynomials $P_n^{(c)}$ can be characterized in terms of the following Riemann-Hilbert problem.  We seek a matrix $\bold P_n^{(c)}(z)$ satisfying the following properties.
\begin{enumerate}
\item $\bold P_n^{(c)}(z)$ is analytic on $\C \setminus \R$.
\item For any real $x$,
\begin{equation}
 \bold P_n^{(c)}(x)=\bold P_n^{(c)}(x)j_c(x),\quad j_c(x)=
\begin{pmatrix}
1 & w(x) \\
0 & 1 
\end{pmatrix}.
\end{equation}
\item As $z\to\infty,$
\begin{equation}
\bold P_n^{(c)}(z)\sim \left( I+\frac {\bold P^{(c)}_1}{z}+\frac {\bold P^{(c)}_2}{z^2}+\ldots\right)
\begin{pmatrix}
z^n & 0 \\
0 & z^{-n} 
\end{pmatrix},
\end{equation}
where $\bold P^{(c)}_k,\; k=1,2,\dots$, are some constant $2\times
  2$ matrices.
\end{enumerate}
This problem has the unique solution
\begin{equation}
 \bold P_n^{(c)}(x)=\begin{pmatrix}
P_n^{(c)}(z) & \frac{1}{2\pi i}\int_{\R}
\frac{P_n^{(c)}(u)\,w_n(u)\,du}{u-z} \\
-\frac{2\pi i}{h_{n-1}}P_{n-1}(z) & -\frac{1}{h_{n-1}}\int_{\R}
\frac{P_{n-1}^{(c)}(u)\,w_n(u)\,du}{u-z}
\end{pmatrix}.
\end{equation}
The normalizing constants $h_n^{(c)}$ can be found as
\begin{equation}
h_n^{(c)}=-2\pi i\left[\bold P^{(c)}_1\right]_{12},\quad (h_{n-1}^{(c)})^{-1}=-\frac{\left[\bold P^{(c)}_1\right]_{21}}{2\pi i}\,.
\end{equation}

We can make a series of transformations to $\bold P_n^{(c)}$ to arrive at a small norm problem. Define $\bold T_n^{(c)}$ from the equation
\begin{equation}
\bold P^{(c)}_n(z)=e^{\frac{nl}{2}\sigma_3}\bold T^{(c)}_n(z)e^{n(g(z)-\frac{l}{2})\sigma_3},
\end{equation}
 and $\bold S_{n}^{(c)}$ as
 \begin{equation}
\begin{aligned}
&\bold S_n^{(c)}(z)=\left\{
\begin{aligned}
&\bold T_n^{(c)}(z)j_+(z)^{-1} \quad \textrm{for} \quad z\in \{(-b,b) \times (0,i\ep)\}  \\
&\bold T_n^{(c)}(z)j_-(z)\quad \textrm{for} \quad z\in \{(-b,b) \times (0,-i\ep)\}  \\
&\bold T_n^{(c)}(z)  \quad \textrm{otherwise}.
\end{aligned}\right. \\ 
\end{aligned}
\end{equation}
Then the matrix $\bold X_n^{(c)}$ can be defined as
\begin{equation}
\bold X_n^{(c)}(z)=\left\{
\begin{aligned}
&\bold S_n^{(c)}(z) \bold M(z)^{-1} \quad \textrm{for} \ z \ \textrm{outside the disks }  D(-b, \ep), \ D(b, \ep)\ \\
&\bold S_n^{(c)}(z) \bold U(z)^{-1} \quad \textrm{for} \ z \ \textrm{inside the disks }  D(-b, \ep), \ D(b, \ep)\,. \\
\end{aligned}\right.
\end{equation}
The jump matrices for $\bold X_n^{(c)}(z)$ are exponentially close to those for $\bold X_n(z)$, and therefore, by (\ref{tt17}) and (\ref{tt18}), $\bold X_n^{(c)}(z)$ and $\bold X_n(z)$ are exponentially close to eachother.  We are interested in the off diagonal terms of the matrix $\bold X_1^{(c)}$, where 
\begin{equation}
\bold X_n^{(c)}(z)=I+\frac{\bold X_1^{(c)}}{z}+O(z^{-2}).
\end{equation}
One can easily see that
\begin{equation}
\begin{aligned}
\left[\bold X_1^{(c)} \right]_{12}&=\left[\bold P_1^{(c)}\right]_{12}e^{-nl}-[\bold M_1]_{12}\, \qquad\qquad \left[\bold X_1^{(c)}\right]_{21}&=\left[\bold P_1^{(c)}\right]_{21}e^{nl}-[\bold M_1]_{21} \\ 
&=-\frac{h_n^{(c)}}{2\pi i} (\pi^2 a e)^n -\frac{i}{\pi \sqrt{a}}\,, \quad &=-\frac{2\pi i}{h_{n-1}^{(c)}(\pi^2 a e)^n}+\frac{i}{\pi \sqrt{a}}\,.
\end{aligned}
\end{equation}
The constants $h_n^{c}$ and $h_{n-1}^{(c)}$ are known exactly:
\begin{equation}
h_{n}^{(c)}=\frac{n! \sqrt{2\pi}}{(\sqrt{na} \pi)^{2n+1}}\,, \qquad\qquad h_{n-1}^{(c)}=\frac{(n-1)! \sqrt{2\pi}}{(\sqrt{na} \pi)^{2n-1}}.
\end{equation}
It follows that
\begin{equation}
\left[\bold X_1^{(c)} \right]_{12}=-\frac{i}{\pi \sqrt{a}} \left(1-\left(\frac{e}{n}\right)^n \frac{n!}{\sqrt{2\pi n}}\right)\,, \quad \left[\bold X_1^{(c)} \right]_{21}=\frac{i}{\pi \sqrt{a}}\left(1-\left(\frac{n}{e}\right)^n\frac{\sqrt{2\pi}}{\sqrt{n} (n-1)!}\right).
\end{equation}
Applying Stirlings formula, we find that
\begin{equation}
\begin{aligned}
\left[{\bf X}_1^{(c)}\right]_{12}=\frac{i}{\pi \sqrt{a}} \left(\frac{1}{12n}+\frac{1}{288n^2} -\frac{139}{51840 n^3} +O(n^{-4})\right), \\
\left[{\bf X}_1^{(c)} \right]_{21}=\frac{i}{\pi \sqrt{a}} \left(\frac{1}{12n}-\frac{1}{288n^2} -\frac{139}{51840 n^3} +O(n^{-4})\right) .
\end{aligned}
\end{equation}
 
 Let us now return to the discrete system of orthogonal polynomials.  The normalizing constants are given as
 \begin{equation}
 \begin{aligned}
h_{n,n}&=\frac{2}{\sqrt{a}}\left(\frac{1}{\pi^2 a e}\right)^n\big(1-[\bold X_1]_{12}\pi i\sqrt{a}\big)  \,, \\
h_{n,n-1}^{-1}&=\frac{1}{2\sqrt{a}\pi^2}\left({\pi^2 a e}\right)^n\big(1+[\bold X_1]_{21}\pi i\sqrt{a}\big)  \,.
\end{aligned}
\end{equation}
Since $\bold X_n$ and $\bold X_n^{(c)}$ are exponentially close to eachother, we may use the above expansion for $\bold X_1$, obtaining
  \begin{equation}
 \begin{aligned}
h_{n,n}&=\frac{2}{\sqrt{a}}\left(\frac{1}{\pi^2 a e}\right)^n \left(1+\frac{1}{12n}+\frac{1}{288n^2} -\frac{139}{51840 n^3} +O(n^{-4})\right)  \,, \\
h_{n,n-1}^{-1}&=\frac{1}{2\sqrt{a}\pi^2}\left({\pi^2 a e}\right)^n\left(1-\frac{1}{12n}+\frac{1}{288n^2} +\frac{139}{51840 n^3} +O(n^{-4})\right) \,.
\end{aligned}
\end{equation}

Let us also note that in this asymptotic regime, by a similar argument, we find that the recurrence coefficients $A_{n,k}^{(\al)}(a)$ are exponentially close to zero, as they vanish for Hermite polynomials.

\section{Deformation equations for orthogonal polynomials}\label{defOP}
In this Appendix, we prove the deformation equations (\ref{dope9}) and (\ref{n4b}).  These equations are in fact quite general, and we present the proof for a general class of orthogonal polynomials.
Let $\{p_k(x)\}_{k=0}^\infty$ be a system of monic polynomials satisfying the orthogonality condition
\begin{equation}\label{def1}
\int_\R p_k(x)p_j(x)e^{-ax^2} d\mu(x)=h_k \de_{jk} ,
\end{equation}
where $d\mu(x)$ is any measure on $\R$ such that the system of orthogonal polynomials exists.  We consider deformations of this system with respect to the parameter $a$.  Let us write the three term recurrence equation, explicitly noting the dependence of each recurrence coefficient on the parameter $a$:
\begin{equation}\label{def2}
xp_k(x)=p_{k+1}(x)+A_k(a)p_k(x)+B_k(a)p_{k-1}(x)\,, \quad B_k(a)=\frac{h_k(a)}{h_{k-1}(a)}\,.
\end{equation}
Notice that, since the polynomials $p_k$ are monic, $\frac{\d}{\d a} p_k(x)$ is a polynomial of degree strictly less than $k$, and thus its integral against $p_k(x) e^{-ax^2} d\mu(x)$ is zero.  Thus if we differentiate (\ref{def1}) with respect to $a$ in the case $j=k$, apply the three term recurrence twice and integrate, we obtain
\begin{equation}\label{def3}
h_k'(a)=-h_k \big(B_{k+1}+A_{k}+B_k\big)\,,
\end{equation}
or equivalently
\begin{equation}\label{def3a}
\frac{\d}{\d a} \log h_k = -A_k- \frac{h_{k+1}}{h_k}- \frac{h_{k}}{h_{k-1}}\,,
\end{equation}
where we have suppressed the notation which explicitly indicates dependence on $a$.

Let us use $c_{k,j}$ to denote the coefficient of the $x^j$ term in the polynomial $p_k(x)$, so that
\begin{equation}\label{def4}
p_k(x)=x^k+c_{k,k-1} x^{n-1} +c_{k,k-2} x^{n-2}+ \cdots.
\end{equation}
These coefficients depend on the parameter $a$, and by matching the  coefficients of the $x^k$ term in (\ref{def2}), we see that
\begin{equation}\label{def5}
A_k(a)=c_{k,k-1}-c_{k+1,k}\,.
\end{equation}
To arrive at a deformation equation for $A_k$ consider (\ref{def1}) with $j=k-1$.  Differentiating with respect to $a$ and disregarding the term for which the integral vanishes gives
\begin{equation}\label{def6}
\int_\R \left[\frac{\d}{\d a} \big(p_k(x)\big) p_{k-1}(x) -x^2 p_{k-1}(x)p_k(x)\right]e^{-ax^2} d\mu(x)=0.
\end{equation}
Applying the three term recursion twice and integrating, we obtain
\begin{equation}\label{def7}
\left(\frac{\d}{\d a} c_{k,k-1}\right) h_{k-1}=A_kh_k+A_{k-1}B_k h_{k-1}\,.
\end{equation}
Combining (\ref{def5}) with (\ref{def7}) both as it is written and with $k \mapsto k+1$, we find
\begin{equation}\label{def8}
A_k'(a)=B_k\big(A_k+A_{k-1}\big) - B_{k+1}\big(A_{k+1}+A_{k}\big).
\end{equation}

We now use (\ref{def3}) and (\ref{def8}) to integrate (\ref{def3}) once more, obtaining
\begin{equation}\label{def9}
\frac{\d^2}{\d a^2} \log h_k=I_{k+1}-I_k\,,
\end{equation}
where
\begin{equation}\label{def10}
I_k=B_k\bigg(B_{k+1}+B_{k-1}+\big(A_{k-1}+A_k\big)^2\bigg)\,.
\end{equation}
It follows that the sum
\begin{equation}\label{def11}
\sum_{k=0}^{n-1} \frac{\d^2}{\d a^2} \log h_k
\end{equation}
telescopes and  its value is $I_n-I_0$.  But $I_0=0$, and thus the sum (\ref{def11}) is simply $I_n$.  After a change of variable, this proves (\ref{dope9}).

We now prove (\ref{n4b}).  In the case that the measure of orthogonality is even, the recurrence coefficients $A_k$ vanish, and we have
\begin{equation}\label{def12}
I_k=B_kB_{k+1}+B_kB_{k-1}\,
\end{equation}
and
\begin{equation}\label{def13}
\sum_{k=0}^{N-1} \frac{\d^2}{\d a^2} \log h_{2k}=\sum_{k=0}^{N-1} B_{2k+1}B_{2k+2}-B_{2k}B_{2k-1}\,; \quad \sum_{k=0}^{n-1} \frac{\d^2}{\d a^2} \log h_{2k+1}=\sum_{k=0}^{n-1} B_{2k+2}B_{2k+3}-B_{2k}B_{2k+1}
\end{equation}
which are again telescoping sums, and we obtain (\ref{n4b}) after a change of variables.

\end{document}